\documentclass[12pt, onecolumn]{IEEEtran}

\usepackage{algorithm}
\usepackage{algorithmic}

\usepackage{times}
\usepackage{algorithm, algorithmic}
\usepackage{amsfonts}
\usepackage{amsmath}
\usepackage{amssymb}
\usepackage{amsthm}
\usepackage{latexsym}
\usepackage{color}
\usepackage{graphics}
\usepackage{enumerate}
\usepackage{amstext}
\usepackage{url}
\usepackage{epsfig}
\usepackage{subfigure}

\DeclareMathOperator{\sgn}{sgn}

\newtheorem{theorem}{Theorem}

\newtheorem{lemma}{Lemma}

\newtheorem*{remark*}{Remark}

\newcommand{\R}{\mathbb{R}}

\newcommand{\T}{{\mathcal T}}

\newcommand{\ignore}[1]{}

\newcommand{\PP}{\mathcal{P}}
\newcommand{\RR}{{\mathcal R}}

\newcommand{\Gammad}{\Gamma_{\textrm{d}}}
\newcommand{\Gammar}{\Gamma_{\textrm{r}}}
\newcommand{\Omegad}{\Omega_{\textrm{d}}}
\newcommand{\Omegar}{\Omega_{\textrm{r}}}
\newcommand{\Gammao}{\Gamma_0}
\newcommand{\OmegaOd}{\OmegaO_{\textrm{d}}}
\newcommand{\OmegaOr}{\OmegaO_{\textrm{r}}}
\newcommand{\OmegaO}{\Phi}

\newcommand{\EE}{E^{\ast}}
\newcommand{\llambda}{\gamma}
\newcommand{\DD}{D}

\begin{document}
%
\title{Low-rank Matrix Recovery from Errors and Erasures}


%
\author{\IEEEauthorblockN{Yudong Chen,
Ali Jalali,
Sujay Sanghavi and
Constantine Caramanis\\}
\IEEEauthorblockA{Department of Electrical and Computer Engineering\\
The University of Texas at Austin, Austin, TX 78712 USA\\
Email: (ydchen, alij, sanghavi and caramanis)@mail.utexas.edu}
}

\maketitle

\begin{abstract}
This paper considers the recovery of a low-rank matrix from an observed version that simultaneously contains both {\em (a) erasures:} most entries are not observed, and {\em (b) errors:} values at a constant fraction of (unknown) locations are arbitrarily corrupted. We provide a new unified performance guarantee on when the natural convex relaxation of minimizing rank plus support succeeds in exact recovery. Our result allows for the simultaneous presence of random and deterministic components in both the error and erasure patterns. On the one hand, corollaries obtained by specializing this one single result in different ways recover (up to poly-log factors) all the existing works in matrix completion, and sparse and low-rank matrix recovery. On the other hand, our results also provide the {\em first guarantees} for (a) recovery when we observe a vanishing fraction of entries of a corrupted matrix, and (b) deterministic matrix completion.
\end{abstract}

\section{Introduction}

Low-rank matrices play a central role in large-scale data analysis and dimensionality reduction. They arise in a variety of application areas, among them Principal Component Analysis (PCA), Multi-dimensional scaling (MDS), Spectral Clustering and related methods, ranking and collaborative filtering, etc. In all these problems, low-rank structure is used to either approximate a general matrix, or to correct for corrupted or missing data.

This paper considers the recovery of a low-rank matrix in the simultaneous presence of {\em (a) erasures:} most elements are not observed, and {\em (b): errors:} among the ones that are observed, a significant fraction at unknown locations are grossly/maliciously corrupted. It is now well recognized that the standard, popular approach to low-rank matrix recovery using SVD as a first step fails spectacularly in this setting \cite{huber}. Low-rank matrix completion, which considers only random erasures (\cite{CanRec,CanTao}) will also fail with even just a few maliciously corrupted entries. In light of this, several recent works have studied an alternate approach based on the natural convex relaxation of minimizing rank plus support. One approach \cite{cspw,tong} provides deterministic/worst case guarantees for the fully observed setting (i.e. only errors). Another avenue \cite{ganesh,CanLiMaWri} provides probabilistic guarantees for the case when the supports of the error and erasure patterns are chosen uniformly at random. Our work provides (often order-wise) stronger guarantees on the performance of this convex formulation, as compared to all of these papers.

We present one main result, and two other theorems. Our main result, Theorem \ref{thm:randomfixed}, is a {\em unified performance guarantee} that allows for the simultaneous presence of both errors and erasures, and deterministic and random support patterns for each. In order/scaling terms, this single result recovers as corollaries all the existing results on low-rank matrix completion \cite{CanRec,CanTao}, worst-case error patterns \cite{cspw}, and random error and erasure patterns \cite{ganesh,CanLiMaWri} up to logarithm factors; we provide detailed comparisons in Section \ref{sec:main}. More significantly, our result goes {\em beyond} the existing literature by providing the first guarantees for random support patterns for the case when the fraction of entries observed vanishes as $n$ (the size of the matrix) grows -- an important regime in many applications, including collaborative filtering. In particular, we show that exact recovery is possible with as few as $\Theta(n \textrm{polylog}(n))$ observed entries, even when a constant fraction of these entries are errors.

Theorem \ref{thm:random} is also a unified guarantee, but with the additional assumption that the {\em signs} of the error matrix are equally likely to be positive or negative. We are now able to show that it is possible to recover the low-rank matrix even when {\em almost all} entries are corrupted. Again, our results go beyond the existing work \cite{ganesh} on this case, because we allow for a vanishing fraction of observations.

Theorem \ref{theorem:mainresult} concentrates on the deterministic/worst-case analysis, providing the first guarantees when there are both errors and erasures. Its specialization to the erasures-only case provides the first deterministic guarantees for low-rank matrix completion (where existing work \cite{CanRec,CanTao} has concentrated on randomly located observations). Specialization to the errors-only case provides an order improvement over the previous deterministic results in \cite{cspw}, and matches the scaling of \cite{tong} but with a simpler proof.

Besides improving on known guarantees, all our results involve several technical innovations beyond existing proofs. Several of these innovations may be of interest in their own right, for other related high-dimensional problems.


\section{Main Contributions}
\label{sec:main}
\subsection{Setup}

\noindent {\bf The problem:} Suppose matrix $C\in\mathbb{R}^{n_1\times n_2}$ is the sum of an underlying low-rank matrix $B^*\in\mathbb{R}^{n_1\times n_2}$ and a sparse ``errors" matrix $A^*\in\mathbb{R}^{n_1\times n_2}$. Neither the number, locations or values of the non-zero entries of $A^*$ are known {\it a priori}; indeed by ``sparse" we just mean that $A^*$ has at least a constant fraction of its entries being 0 -- it is allowed to have a significant fraction of its entries being non-zero as well. We consider the following problem: suppose we only observe a subset $\OmegaO \subseteq [n_1] \times [n_2]$ of the entries of $C$; the remaining entries are erased. When and how can we \underline{exactly} recover $B^*$ (and, by simple implication, the entries of $A^*$ that are in $\OmegaO$)?\\

\noindent {\bf The Algorithm:} In this paper we are interested in the performance of the following convex program
\begin{equation}
\begin{aligned}
(\hat{A},\hat{B})=\arg&\min_{A,B}\qquad \gamma\|A\|_{1}+\|B\|_*\\
&\quad\text{s.t.}\qquad \PP_{\OmegaO}\left(A+B\right)=\PP_{\OmegaO}\left(C\right),
\end{aligned}
\label{eq:opt_problem}
\end{equation}
where the notation is that for any matrix $M$, $\|M\|_* = \sum_i \sigma_i(M) $ is the nuclear norm, defined to be the sum of the singular values of the matrix, $\|M\|_{1} = \sum_{i,j} |a_{ij}|$ is the elementwise $\ell_1$ norm, and $\PP_{\OmegaO}(M)$ is the matrix obtained by setting the entries of $M$ that are outside the observed set $\OmegaO$ to zero. Intuitively, the nuclear norm acts as a convex surrogate for the rank of a matrix \cite{fazel}, and the $\ell_1$ norm as a convex surrogate for its sparsity. Here $\llambda$ is a parameter that trades off between these two elements of the cost function, and our results below specify how it should be chosen. As noted earlier, this program has appeared previously in \cite{CanLiMaWri,cspw}.\\

\noindent {\bf Incoherence:} We are interested in characterizing when the optimum of \eqref{eq:opt_problem} recovers the underlying (observed) truth, i.e., when $(\PP_{\OmegaO}(\hat{A}),\hat{B}) = \left(\PP_{\OmegaO}\left(A^*\right),B^*\right)$. Clearly, not all low-rank matrices $B^*$ can be recovered exactly; in particular, if $B^*$ is both low-rank {\em and} sparse, it would be impossible to unambiguously identify it from an added sparse matrix. To prevent such a scenario, we follow the approach taken in the recent work \cite{cspw,CanLiMaWri,CanRec,CanTao,Gro} and define {\em incoherence} parameters for $B^*$. Suppose the matrix $B^*$ with rank $r\leq\min\left(n_1,n_2\right)$ has  singular value decomposition $U\Sigma V^\top$, where $U\in\mathbb{R}^{n_1\times r}$, $V\in\mathbb{R}^{n_2\times r}$ and $\Sigma\in\mathbb{R}^{r\times r}$. We say a given matrix $B^*$ is \textbf{$\mu$-incoherent} for some $\mu\in\left[1,\frac{\max\left(n_1,n_2\right)}{r}\right]$ if 
\begin{equation}
\begin{aligned}
\max_{i}\|U^\top\mathbf{e}_{i}\|\leq\sqrt{\frac{\mu r}{n_1}}\qquad &\qquad\max_{j}\|V^\top\mathbf{e}_{j}\|&\leq\sqrt{\frac{\mu r}{n_2}}\\
\|UV^\top\|_{\infty}&\leq\sqrt{\frac{\mu r}{n_1 n_2}},
\end{aligned}
\nonumber
\end{equation}
where, $\mathbf{e}_{i}$'s are standard basis vectors with proper length, and $\|\cdot\|$ represents the $2$-norm of the vector. Notice that all our results in the following subsections only depend on the product of $\mu$ and $r$.\\

\subsection{Unified Guarantee}

Our first main result is a unified guarantee that allows for the simultaneous presence of random and adversarial patterns, for both errors and erasures. As mentioned in the introduction, this recovers all existing results in matrix completion, and sparse and low-rank matrix decomposition, up to constants or $\log$ factors. We now define three bounding quantities: $p_0, \tau$ and $d$.

Let $\OmegaOd$ be any (i.e. deterministic) set of observed entries, and additionally let $\OmegaOr$ be a randomly chosen set such that each entry is in $\OmegaOr$ with probability {\em at least} $p_0$. Thus, the overall set of observed entries is $\OmegaO=\OmegaOr \cap\OmegaOd$, the \emph{intersection} of the two sets. Let $\Omega=\Omega_{\textrm{r}}\cup\Omega_{\textrm{d}}$ be the support of $A^*$, again composed of the \emph{union} of a deterministic component $\Omega_{\textrm{d}}$, and a random component $\Omega_{\textrm{r}}$ generated by having each entry be in $\Omega_{\textrm{r}}$ independently with probability {\em at most} $\tau$. Finally, consider the union $\OmegaOd^{c}\cup\Omega_{\textrm{d}}$ of all deterministic errors and erasures, and let $d$ be an upper bound on the maximum number of entries this set has in any row, or in any column.

\begin{theorem}
[Unified Guarantee]
Set $n=\min\{n_1,n_2\}$. There exist universal constants $C$, $\rho_{r}$, $\rho_{s}$ and $\rho_{d}$ -- each independent of $n$, $\mu$ and $r$ -- such that, with probability greater than $1-Cn^{-10}$, the unique optimal solution of \eqref{eq:opt_problem} with tradeoff parameter $\llambda=\frac{1}{32\sqrt{p_{0}(d+1)n}}$
is equal to $\left(\PP_{\OmegaO}(A^{\ast}),\; B^{\ast}\right)$
provided that
\begin{eqnarray*}
p_{0} & \ge &  \rho_{\textrm{r}} \frac{\mu r\log^{6}n}{n}\\
\tau & \le & \rho_{s}\\
d & \le & \rho_{\textrm{d}} \frac{n}{\mu r}\cdot\frac{p_{0}^2}{\log^4 n}
\end{eqnarray*}
\label{thm:randomfixed}
\end{theorem}
\begin{remark*}
(a) The conclusion of the theorem holds for a range of values of $\gamma$. We have chosen one of these valid values. (b) Note that the above theorem treats errors and erasures differently. Treating erasures as errors by filling missing entries with random $\pm 1$ and applying Theorem \ref{thm:random} leads to a weaker result, in particular, $p_0=\Omega\left( \sqrt{\frac{\mu r \log^6 n}{n}} \right)$.
\end{remark*}


\textbf{Comparison with previous work.}
Recovery from deterministic errors was first studied in \cite{cspw,cspw1}, which stipulate $d=O\left( \sqrt{\frac{n}{\mu r}} \right)$. Our theorem improves this bound to $d=O\left(  \frac{n}{\mu r \log^4 n} \right)$. In section \ref{sub:determ_guarantee}, we provide a more refined analysis for the deterministic case, which gives $d=O\left( \frac{n}{\mu r} \right)$. As this manuscript was being prepared, we learned of an independent investigation of the deterministic case \cite{tong}, which gives similar guarantees. Our results also handle the case of partial observations, which has not been discussed before \cite{cspw,cspw1,tong}.

Randomly located errors and erasures have been studied in \cite{CanLiMaWri}. Their guarantees require that $\tau=O(1)$, and $p_0 = \Omega(1)$. Our theorem provides stronger results, allowing $p_0$ to be vanishingly small, in particular, $\Theta\left( \frac{\mu r\log^6 n}{n} \right)$ when there is no additional deterministic component (i.e. $d=0$). After the publication of the conference version of this paper, we learned about \cite{xiaodongli}. They also deal with random errors and erasures, but under a different observation model (sampling with replacement), and have scaling results comparable to ours.

Previous work in low-rank matrix completion deals with the case when there are no errors or deterministic erasures (i.e., $d, \tau=0$). For this problem, our theorem matches the best existing bound $p_{0} = O\left(\frac{\mu r \log^2 n}{n}\right)$ \cite{CanTao, Gro, Recht} up to logarithm factors. Our theorem also provides the first guarantee for deterministic matrix completion under potentially adversarial erasures.

One prominent feature of our guarantees is that we allow adversarial and random erasures/errors to exist \emph{simultaneously}. To the best of our knowledge, this is the first such result in low-rank matrix recovery/robust PCA.

\subsection{Improved Guarantee for Errors with Random Sign}
If we further assume that the errors in the entries in $\Omegar\backslash\Omegad$ have random signs, then one can recover from an overwhelming fraction of corruptions.
\begin{theorem}
[Improved Guarantee for Errors with Random Sign]
Under the same setup of Theorem \ref{thm:randomfixed}, further assume that the signs of $A^{\ast}$ in $\Omegar\backslash\Omegad$ are symmetric $\pm 1$ Bernoulli random
variables independent of all others. Then there exist absolute constants $C$, $\rho_{\textrm{r}}$ and $\rho_{\textrm{d}}$ independent of $n$, $\mu$ and $r$
such that, with probability at least $1-Cn^{-10}$, the unique optimal solution of \eqref{eq:opt_problem}
with tradeoff parameter $\llambda=\frac{1}{32\sqrt{p_{0}(d+1)n}}$ is equal to $\left(\PP_{\OmegaO}(A^{\ast}),\; B^{\ast}\right)$
provided that
\begin{eqnarray*}
p_{0}(1-\tau)^{2} & \ge &  \rho_{\textrm{r}} \frac{\mu r\log^{6}n}{n}\\
d & \le & \rho_{\textrm{d}} \frac{n}{\mu r}\cdot\frac{p_{0}^2(1-\tau)^2}{\log^4 n}
\end{eqnarray*}
\label{thm:random}
\end{theorem}
\begin{remark*}
Note that $\tau$ may be arbitrary close to $1$ for large $n$. One interesting observation is that $p_0$ can approach zero faster than $1-\tau$; this agrees with the intuition that correcting erasures with known locations is easier than correcting errors with unknown locations.
\end{remark*}

\textbf{Comparison with previous work } Dense errors with random locations and signs were considered in \cite{ganesh}. They show that $\tau$ can be a constant arbitrarily close to $1$ provided that all entries are observed and $n$ is sufficiently large. Our theorem provides stronger results by again requiring only a vanishingly small fraction of entries to be observed and in particular  $p_0 = \Theta\left( \frac{\log^4 n}{n} \right)$. Moreover, Theorem \ref{thm:random} gives explicit scaling between $\tau$ and $n$ as $\tau=O\left( 1-\sqrt{\frac{\log^4 n}{n}}\right)$, with $\llambda$ independent of the usually unknown quantity $\tau$. In contrast, \cite{ganesh} requires $ \tau \le f(n)$ for some unknown function $f(\cdot)$ and uses a $\tau$-dependent $\llambda$.

\subsection{Improved Deterministic Guarantee}
\label{sub:determ_guarantee}
\noindent Our second main result deals with the case where the errors and erasures are arbitrary. As discussed in \cite{cspw}, for exact recovery, the error matrix $A^*$ needs to be not only sparse but also "spread out", i.e. to not have any row or column with too many non-zero entries. The same holds for unobserved entries. Correspondingly, we require the following: (i) there are at most $d$ errors and erasures on each row/column, and, (ii) $\|M\|\leq\eta d\|M\|_{\infty}$ for any matrix $M$ that is supported on the set of corrupted entries and unobserved entries; here $\|M^*\|=\sigma_{\max}(M^*)$ is the largest singular value of $M$ and $\|M\|_{\infty} = \max_{i,j} |M_{i,j}|$ is the element-wise maximum magnitude of the elements of the matrix. Note that by \cite[Proposition 3]{cspw}, we can always take $\eta\le 1$. Also, let $\alpha=\sqrt{\frac{\mu r d}{n_1}}+\sqrt{\frac{\mu r d}{n_2}}+\sqrt{\frac{\mu r d}{\max(n_1,n_2)}}$.\\

\begin{theorem}
[Improved Deterministic Guarantee]
For tradeoff parameter \small$\llambda\in\left[\frac{1}{1-2\alpha}\sqrt{\frac{\mu r}{n_1n_2}},\frac{1-\alpha}{\eta d}-\sqrt{\frac{\mu r}{n_1n_2}}\right]$\normalsize, suppose \small$$\sqrt{\frac{\mu r d}{\min(n_1,n_2)}}\!\left(\!\!1\!+\!\sqrt{\frac{\min(n_1,n_2)}{\max(n_1,n_2)}}\!+\!\eta\sqrt{\frac{ d}{\max(n_1,n_2)}}\right)\!\!\leq\!\frac{1}{2}.$$\normalsize Then, the solution to the problem~(\ref{eq:opt_problem}) is unique and equal to $(\PP_{\OmegaO}(A^*),B^*)$.\\
\label{theorem:mainresult}
\end{theorem}
\begin{remark*}
(a) Notice that we have $\sqrt{d}$ in the bound while \cite{cspw} has $d$ in their bound. This improvement is achieved by a different construction of dual certificate presented in this paper. (b) If $\eta d\sqrt{\frac{\mu r}{\min(n_1,n_2)}}\leq\frac{1}{6}$ (the condition provided for exact recovery in \cite{cspw}) is satisfied then the condition of Theorem~\ref{theorem:mainresult} is satisfied as well. This shows that our result is an improvement to the result in \cite{cspw} in the sense that this result guarantees the recovery of a larger set of matrices $A^*$ and $B^*$. Moreover, this bound implies that $n$ (for square matrices) should scale with $dr$, which is another improvement compared to the $d^2r$ scaling in \cite{cspw}. (c) We construct the dual certificate by the method of least squares (first used in \cite{CanRec} in a different setting) with tighter bounding. This theorem provides the same scaling result for $d$, $r$ and $n$ as that in the recent manuscript \cite{tong}. However, our assumptions are closer to existing ones in matrix completion and sparse and low-rank decomposition papers \cite{CanRec,CanTao, cspw, CanLiMaWri}.
\end{remark*}

\newcommand{\Ai}{\mathcal{A}_i}
\newcommand{\Ak}{\mathcal{A}_k}
\newcommand{\Sk}{\mathcal{S}_k}
\newcommand{\Si}{\mathcal{S}_i}
\newcommand{\Bk}{\mathcal{B}_k}
\newcommand{\Tk}{\mathcal{T}_k}
\newcommand{\XX}{\mathcal{X}}

\newcommand{\Gammak}{\Gamma^{(k)}}
\newcommand{\Gammai}{\Gamma^{(i)}}
\newcommand{\Gammaki}{\Gamma^{(k-i)}}

\newcommand{\Omegak}{\Omega^{(k)}}
\newcommand{\Omegai}{\Omega^{(i)}}
\newcommand{\Omegaone}{\Omega^{(1)}}
\newcommand{\Omegatwo}{\Omega^{(2)}}
\newcommand{\Omegako}{\Omega^{(k_0)}}

\newcommand{\OmegaOk}{\OmegaO^{(k)}}
\newcommand{\OmegaOi}{\OmegaO^{(i)}}

\section{Proof Theorem \ref{thm:randomfixed} and \ref{thm:random}}

In this section we prove our unified guarantees. The main roadmap is along the same lines of those in the low-rank matrix recovery literature \cite{CanRec, CanLiMaWri, Gro}; it consists of providing a dual matrix $Q$ that certifies the optimality of $(\PP_{\OmegaO}(A^\ast), B^\ast)$ to the convex program \eqref{eq:opt_problem}. In spite of this high level similarity, challenges arise because of the denseness of erasures/errors as well as the simultaneous presence of deterministic and random components. This requires a number of innovative intermediate results and a new construction of the dual certificate $Q$. We will point out how our analysis departs from previous works when we construct the dual certificate in section \ref{sub:dualcert}.

Before proceeding, we need to introduce some additional notation.   Define the support of $A^*$ as $\Omega = \{(i,j): A^*_{i,j} \neq 0 \}$. Let $\Gamma = \OmegaO \backslash \Omega$ be the set of  entries that are observed \emph{and} clean, then $\Gamma^c$ is the set of entries that are corrupted \emph{or} unobserved.
 Also, let $\Gammar=\OmegaOr \backslash \Omegar$
be the set of \emph{random} observed clean entries, and $\Gammad$ the set of \emph{deterministic} observed clean entries; so $\Gamma = \Gammar \cap \Gammad$. The projections $\PP_{\Gamma}$, $\PP_{\Gamma^{c}}$, $\PP_{\Gammar}$, and $\PP_{\Gammar^{c}}$ are defined similarly to $\PP_{\OmegaO}$.
Set $\EE:=\PP_{\OmegaO}\left(\sgn(A^{\ast})\right)$,
where $\sgn(\cdot)$ is the element-wise signum function. For an entry set $\Omega_{0}$, we write $\Omega_{0}\sim \textrm{Ber}(p)$ if $\Omega_{0}$ contains each entry with probability $p$, independent of all others; therefore $\OmegaOr\sim\textrm{Ber}(p_{0})$, $\Omegar\sim\textrm{Ber}(\tau)$, and $\Gammar\sim\textrm{Ber}(p_{0}(1-\tau))$. We also define a sub-space $\T$ of the span of all matrices that share either the same column space or the same row space as $B^*$:
\begin{equation*}
\T ~ = ~ \left\{UX^{\top}+YV^{\top}:X \in \R^{n_{2}\times r}, Y \in \R^{n_{1}\times r}\right\}.
\end{equation*}
For any matrix $M \in\R^{n_{1} \times n_{2}}$, we can define its {\em orthogonal projection} to the space $\T$ as follows:
\begin{equation*}
\PP_{\T}\left(M\right)=UU^{\top}M+MVV^{\top}-UU^{\top}MVV^{\top}.
\end{equation*}
We also define the projections onto $\T^{\perp}$, the complement orthogonal space of $\T$, as follows:
\begin{equation*}
\PP_{\T^{\perp}}\left(M\right) ~ = ~ M - \PP_\T(M).
\end{equation*}
In the sequel, we use $C$, $C'$ and $C''$ to denote unspecified positive constants, which might differ from place to place; by \emph{with high probability} we mean with probability at least $1-C\min\{n_1,n_2\}^{-10}$. For simplicity, we only prove the case of square matrices ($n_1=n_2=n$). All the proofs extend to the general case by replacing $n$ by $\min\{n_1,n_2\}$. The proof has five steps. We elaborate each of these steps in the next five sub-sections.

\subsection{Step 1: Sign Pattern Derandomization}
\label{sub:derand}
Following \cite{CanLiMaWri}, the first step is to observe that it suffices to prove Theorem \ref{thm:random}, which assumes random signed errors in $\Omegar\backslash\Omegad$. The guarantee under arbitrary signed errors in Theorem \ref{thm:randomfixed} follows automatically from Theorem \ref{thm:random} using a derandomization and elimination argument. This is given in the following lemma, which is a straightforward generalization of \cite[Theorem 2.2 and 2.3]{CanLiMaWri}.
\begin{lemma}
Suppose $B^*$ obeys the conditions of Theorem \ref{thm:randomfixed}. If the convex program \eqref{eq:opt_problem} recovers $B^*$ with high probability in the model where $\Omegar \sim \textrm{Ber}(2\tau)$ and the signs of $A^*$ in $\Omegar\backslash\Omegad$ have random signs, then it also recovers $B^*$ with at least the same probability in the model where $\Omegar \sim \textrm{Ber}(\tau)$ and the signs are arbitrarily fixed.
\end{lemma}
The basic idea of the proof is that, as long as $\tau$ is not too large, a fixed-signed error matrix $\PP_{\Omegar\backslash\Omegad}(A^\ast)$ can be viewed as the trimmed version of a random signed $\PP_{\bar{\Omegar}\backslash\Omegad}(\bar{A}^\ast)$ with half of its entries set to zero; moreover, successful recovery under $A^\ast$ is guaranteed by that under $\bar{A}^\ast$, as the latter is a harder problem. We refer the readers to \cite[Theorem 2.2 and 2.3]{CanLiMaWri} for the rigorous proof of this argument. Proceeding under the random-sign assumption makes it easier to construct the dual certificate $Q$. The next four steps are thus devoted to the proof of Theorem \ref{thm:random}.

\subsection{Step 2: Invertibility under corruptions and erasures}
A necessary condition for exact recovery is that the set of uncorrupted and un-erased entries $\Gamma = \Gammar \cap \Gammad$ should uniquely identify matrices in the set $\T$, so we need to show that the operator $\PP_{\T}\PP_{\Gamma}\PP_{\T}$ is invertible on $\T$. This step is quite standard in the literature of low-rank matrix completion and decomposition, but in our case requires a different proof. In fact, invertibility follows from the following stronger result.
\begin{lemma}
\label{lem:OP}
Suppose $\Omega_{0}$ is a set of indices obeying $\Omega_{0}\sim$Ber$(p)$, and $\Gammad$ satisfies $d\le\frac{\rho_d n}{\mu r}$. Then with high probability, we have
\[
\left\Vert p^{-1}\PP_{\T}\PP_{\Omega_{0}\cap\Gammad}\PP_{\T}-\PP_{\T}\right\Vert \le\frac{1}{3}
\]
provided $p\ge C\frac{\mu r \log n}{n}$.
\end{lemma}
Invertibility follows from specializing $\Omega_{0} = \Gammar$. The lemma is stated in terms of a generic entry set $\Omega_{0}$ because it is invoked again elsewhere. Notice that this lemma is a generalization of \cite[Theorem 4.1]{CanRec}, as $\Omega_{0}\cap \Gammad$ involves both random and deterministic components. The proof is new, utilizing the properties of both components, and is given in the appendix.

\subsection{Step 3: Sufficient Conditions for Optimality}
The next step is to use convex analysis to write down the first-order sub-gradient sufficient condition for $(\PP_{\OmegaO}(A^\ast), B^\ast)$ to be the unique solution to \eqref{eq:opt_problem}. This is given in the following lemma. Recall that we have defined $\EE:=\PP_{\OmegaO}\left(\sgn(A^{\ast})\right)$.
\begin{lemma}
\label{lem:OptimalityCondition}
Suppose $\llambda$, $p_{0}$, $\tau$ and $d$ satisfy the condition in Theorem \ref{thm:random}. Then with high probability $(\PP_{\OmegaO}(A^{\ast}),\; B^{\ast})$ is
the unique solution to \eqref{eq:opt_problem} if there is a dual certificate $Q=\llambda \EE + W$ obeying
\begin{eqnarray}
(a) && \left\Vert \PP_{\T}W-(UV^{\top}-\llambda\PP_{\T}\EE)\right\Vert _{F} \le \frac{\llambda}{\sqrt{n}} \nonumber \\
(b) && \PP_{\Gamma^{c}}W = 0. \nonumber\\
(c) && \left\Vert \PP_{\Gamma}W\right\Vert _{\infty} < \frac{\llambda}{2}\label{eq:QRequirement}\\
(d) && \left\Vert \PP_{\T^{\bot}}W\right\Vert < \frac{1}{4} \nonumber\\
(e) && \left\Vert \llambda \PP_{\T^{\bot}}\EE\right\Vert < \frac{1}{4}. \nonumber
\end{eqnarray}
\end{lemma}
\begin{proof}
Observe that the conditions in the lemma imply $\PP_{\OmegaO^c} (Q) = 0 $, $\left\Vert\PP_{\T}(Q) - UV^\top\right\Vert_F \le \frac{\gamma}{\sqrt{n}}$, $\left\Vert\PP_{\T^\perp}(Q)\right\Vert < \frac{1}{2}$, $\PP_{\Omega}(Q) = \llambda \EE$, and $\left\Vert\PP_{\Gamma}\right\Vert_\infty<\frac{\gamma}{2}$. Consider another feasible solution $(\PP_{\OmegaO}(A^{\ast})+\Delta_{2}, \; B^*+\Delta_{1})$ with $\Delta_{1}\neq0$, $\Delta_{2}\neq0$, and $\PP_{\OmegaO}(\Delta_{1}+\Delta_{2})=0$. Take $G_0\in\T^{\bot}$ and $F_0\in \Gamma$ such that $\Vert G_0\Vert=1$, $\Vert F_0 \Vert_{\infty}=1$, $\left\langle G_0,\; \Delta_1\right\rangle = \left\Vert \PP_{\T^{\bot}}\Delta_1 \right\Vert_*$ and $\left\langle F_0,\; \Delta_2\right\rangle = \left\Vert \PP_{\Gamma} \Delta_2 \right\Vert_1$; such $G_{0}$ and $F_{0}$ exist due to the duality between $\left\Vert\cdot\right\Vert_*$ and $\left\Vert\cdot\right\Vert$, and that between $\left\Vert\cdot\right\Vert_1$ and $\left\Vert\cdot\right\Vert_\infty$. We then have
\begin{eqnarray}
 &  & \left\Vert B^*+\Delta_{1}\right\Vert _{\ast}+\llambda\left\Vert \PP_{\OmegaO}(A^{\ast})+\Delta_{2}\right\Vert _{1}-\left\Vert B^*\right\Vert _{\ast}-\llambda\left\Vert \PP_{\OmegaO}(A^{\ast})\right\Vert _{1} \nonumber \\
 & \ge & \left\langle UV^\top+G_{0},\;\Delta_{1}\right\rangle +\llambda\left\langle \EE+F_{0},\;\Delta_{2}\right\rangle \nonumber \\
 & = & \left\langle UV^\top+G_{0}-Q,\;\Delta_{1}\right\rangle +\left\langle \llambda\EE + \llambda F_{0}-Q,\;\Delta_{2}\right\rangle \nonumber \\
 & = & \left\langle G_{0}-\PP_{\T^\bot}(Q)-\left(\PP_{\T}(Q)-UV^\top\right),\;\Delta_{1}\right\rangle +\left\langle \llambda F_{0}-\PP_{\Gamma}(Q),\;\Delta_{2}\right\rangle \nonumber \\
 & \ge & \left\Vert \PP_{\T^{\bot}}\Delta_{1}\right\Vert_{\ast}\left(1-\left\Vert \PP_{\T^\bot}(Q)\right\Vert \right)-\left\Vert \PP_{\T}(Q)-UV^\top \right\Vert _{F}\left\Vert \PP_{\T}\Delta_{1}\right\Vert _{F}+\left\Vert \PP_{\Gamma}\Delta_{2}\right \Vert _{1}\left(\llambda-\left\Vert \PP_{\Gamma}(Q)\right\Vert _{\infty}\right) \label{eq:opt_cond_eq1} \\
 & \ge & \frac{1}{2}\left\Vert \PP_{\T^{\bot}}\Delta_{1}\right\Vert_{\ast}- \frac{\gamma}{\sqrt{n}}\left\Vert \PP_{\T}\Delta_{1}\right\Vert _{F} + \frac{\gamma}{2}\left\Vert \PP_{\Gamma}\Delta_{2}\right \Vert _{1}; \nonumber
\end{eqnarray}
here we use the sub-gradients of $\Vert \cdot \Vert_*$ and $\Vert \cdot \Vert_1$ in the first inequality and Cauchy-Schwarz inequality in \eqref{eq:opt_cond_eq1}. We need to upper-bound $\left\Vert \PP_{\T}\Delta_{1}\right\Vert_{F} $. Notice that w.h.p.
\begin{eqnarray*}
 & & \left\Vert \PP_{\Gamma}\PP_{\T}\Delta_{1}\right\Vert _{F}^{2} \\
 & = & \left\langle \PP_{\T}\Delta_{1},\; \PP_{\T}\PP_{\Gamma}\PP_{\T}\Delta_{1}\right\rangle \\
 & = & \left\langle \PP_{\T}\Delta_{1},\; \PP_{\T}\PP_{\Gamma}\PP_{\T}\Delta_{1}-p_{0}(1-\tau)\PP_{\T}\Delta_{1} +p_{0}(1-\tau)\PP_{\T}\Delta_{1}\right\rangle \\
 & \ge & p_{0}(1-\tau)\left\Vert \PP_{\T}\Delta_{1}\right\Vert _{F}^{2} - \frac{1}{2} p_{0}(1-\tau)\left\Vert \PP_{\T}\Delta_{1}\right\Vert _{F}^{2}\\
 & = & \frac{1}{2}p_{0}(1-\tau)\left\Vert \PP_{\T}\Delta_{1}\right\Vert _{F}^{2};
\end{eqnarray*}
here in the inequality we use Lemma \ref{lem:OP} with $\Omega_{0} = \Gammar$ and $p=p_0(1-\tau)$. It follows that
\begin{eqnarray*}
 & & \left\Vert \PP_{\Gamma}\Delta_{2}\right\Vert _{1}
 \ge  \left\Vert \PP_{\Gamma}\Delta_{2}\right\Vert _{F}
  =  \left\Vert \PP_{\Gamma}\Delta_{1}\right\Vert _{F}\\
 & = & \left\Vert \PP_{\Gamma}\PP_{\T}\Delta_{1}+\PP_{\Gamma}\PP_{\T^{\bot}}\Delta_{1}\right\Vert _{F}\\
 & \ge & \left\Vert \PP_{\Gamma}\PP_{\T}\Delta_{1}\right\Vert _{F}-\left\Vert \PP_{\Gamma}\PP_{\T^{\bot}}\Delta_{1}\right\Vert _{F}\\
 & \ge & \sqrt{\frac{p_{0}(1-\tau)}{2}}\left\Vert \PP_{\T}\Delta_{1}\right\Vert _{F}-\left\Vert \PP_{\T^{\bot}}\Delta_{1}\right\Vert _{F}\\
 & \ge & \sqrt{\frac{4}{n}}\left\Vert \PP_{\T}\Delta_{1}\right\Vert _{F}-\left\Vert \PP_{\T^{\bot}}\Delta_{1}\right\Vert _{\ast},
\end{eqnarray*}
where the last inequality holds under the assumptions in Theorem \ref{thm:random}. Substituting back to \eqref{eq:opt_cond_eq1}, we obtain
\begin{eqnarray*}
 &  & \left\Vert B^*+\Delta_{1}\right\Vert _{\ast}+\llambda\left\Vert \PP_{\OmegaO}(A^{\ast}) +\Delta_{2}\right\Vert _{1}-\left\Vert B^*\right\Vert _{\ast}-\llambda\left\Vert \PP_{\OmegaO}(A^{\ast})\right\Vert _{1}\\
 & \ge & \left\Vert \PP_{\T^{\bot}}\Delta_{1}\right\Vert _{\ast}\left(\frac{1}{2} -\frac{\gamma}{2}\right) + \left\Vert \PP_{\Gamma}\Delta_{2}\right\Vert _{1}\left(\frac{\gamma}{2}-\frac{\gamma}{2}\right)\\
 & \ge & 0,
\end{eqnarray*}
where we use $\llambda<1$. We claim that the above inequality is strict. Suppose it is not, then we must have $\PP_{\T^{\bot}}\Delta_{1}=\PP_{\Gamma}\Delta_{2}=0$. But under the assumptions in Theorem \ref{thm:random}, $\PP_{\T}\PP_{\Gamma}\PP_{\T}$ is invertible by Lemma \ref{lem:OP} and thus $\Gamma^{\bot}\cap \T=\{0\}$, which contradicts $\Delta_{1}\neq0$ and $\Delta_{2}\neq0$.
\end{proof}

\subsection{Step 4: Construction of the Dual Certificate}
\label{sub:dualcert}
We need to show the existence a matrix $W$ obeying the conditions in \eqref{eq:QRequirement} in Lemma \ref{lem:OptimalityCondition}. We will construct $W$ using a variation of the so-called Golfing Scheme \cite{CanLiMaWri,Gro}. Here we briefly explain the idea. Consider the left hand side of condition (a) in \eqref{eq:QRequirement} as the ``error'' of approximating $UV^{\top}-\llambda\PP_{\T}\EE$ by $\PP_{\T}W$; we want the error to be small. First observe that the choice of $W = UV^\top-\llambda\PP_{\T}\EE$ satisfies (a) strictly but violates (b). To enforce (b), one might consider \emph{sampling} according to $\Gamma$, the set of observed clean entries, and define
\[
W_1 = (p_{0}(1-\tau))^{-1}\PP_{\Gamma} \left( UV^\top-\llambda\PP_{\T}\EE \right).
\]
With the choice of $W=W_1$, (b) is satisfied, and one expects the error in (a) is also small because its \emph{expectation} equals $ -\PP_{\T}\PP_{\Gammad^{c}} \left( UV^\top-\llambda\PP_{\T}\EE \right)$, which is small as long as $\PP_{\Gammad^c}$ is a contraction. This intuition is largely true except that the error is still not small enough. To correct this bias, it is natural to compensate by subtracting the remaining error from $W_1$, and then sample again. Indeed, if one sets $W_2=W_1 - (p_{0}(1-\tau))^{-1}\PP_{\Gamma} \left( \PP_{\T} W_1 - ( UV^\top-\llambda\PP_{\T}\EE) \right)$, then $W=W_2$ still satisfies (b), and the error in (a) becomes smaller. By repeating this ``correct and sample'' procedure, the error actually decreases geometrically fast.

This is almost exactly how we are going to construct $Q$; the only modification is that for technical reasons we need to decompose the observed clean entry set $\Gamma$ into independent batches and sample according to a different batch at each step. To this end, we think of $\Omegar^c\sim\textrm{Ber}\left(1-\tau\right)$ as $\cup_{1\le k\le k_{0}}\Omegak$ and $\OmegaOr\sim\textrm{Ber}\left(p_0\right)$ as $\cup_{1\le k\le k_{0}}\OmegaOk$, where the sets $\Omegak\sim\textrm{Ber}(q_1)$ and $\OmegaOk\sim\textrm{Ber}(q_2)$ are independent; here $k_0$ is taken to be $\left\lceil 4\log n\right\rceil $, and $q_1,q_2$ obeys $1-\tau=1-(1-q_1)^{k_{0}}$ and $p_0=1-(1-q_2)^{k_{0}}$.  Observe that $q_1\ge (1-\tau)/k_{0}$ and $q_2\ge p_0/k_0$.  One can verify that $\Omegar$ and $\OmegaOr$ have the same distribution as before. Define $\Gammak=\Omegak\cap\OmegaOk$, which can be considered as the $k$-th batch of (random) observed clean entries; we then have $\Gammak\sim\textrm{Ber}(q)$ with $q:=q_1q_2\ge \frac{p_0(1-\tau)}{k_0^2}\ge C \frac{\mu r \log n}{n}$, where $C$ may become arbitrarily large by selecting $\rho_r$ sufficiently large. Define the operator $\RR_{\Gammak}:\mathbb{R}^{n\times n}\mapsto\mathbb{R}^{n\times n}$
as
\begin{equation*}
\RR_{\Gammak}(M) \triangleq q^{-1}\PP_{\Gammak \cap \Gammad}(M) = \sum_{i,j\in\Gammak\cap \Gammad}q^{-1}M_{i,j}(e_{i}e_{j}^{\top}),
\end{equation*}
which is simply the (properly scaled) projection onto the $k$-th batch of observed clean entries. The matrix $W$ is then constructed as $W = W_{k_{0}}$, where $W_{k_{0}}$ is defined recursively by $W_{0} := 0$ and
\begin{equation*}
W_{k} := W_{k-1}+\RR_{\Gammak}\left(UV^{\top}-\llambda \PP_{\T}\EE-\PP_{\T}W_{k-1}\right),\qquad \textrm{for } k=1,2,\ldots,k_{0}.
\end{equation*}

The previous work \cite{CanLiMaWri} also applies Golfing Scheme, but only to the part of the dual certificate that involves $UV^\top$; for the part that involves $\EE$, they use the method of least squares. We utilize Golfing Scheme for both parts of the certificate. Difficulties arise due to the dependence between $\EE$ and $\Gammak$'s, and a new analysis is needed for the validation of the certificate. This crucial difference allows us to go beyond \cite{CanLiMaWri} and handle a vanishing fraction of observations and/or clean entries.

\subsection{Step 5: Validity of the Dual Certificate}

It remains to show that $Q$ satisfies all the constraints
in the optimality condition \eqref{eq:QRequirement} simultaneously. The equality (b) is immediate by the construction of $Q$ and $W$.  To prove the inequalities, one observes that if we denote the $k$-th step error as $\DD_{k} :=  UV^{\top}-\llambda \PP_{\T}\EE-\PP_{\T}W_{k}$, then $D_k$ satisfies the following recursion
\begin{eqnarray}
\DD_{k} & = & UV^{\top}-\llambda \PP_{\T}\EE-\PP_{\T}W_{k}\nonumber \\
 & = & (\PP_{\T}-\PP_{\T}\RR_{\Gammak}\PP_{\T})(UV^{\top}-\llambda \PP_{\T}\EE-\PP_{\T}W_{k-1})\nonumber \\
 & = & (\PP_{\T}-\PP_{\T}\RR_{\Gammak}\PP_{\T})\DD_{k-1},
 \label{eq:stepError}
 \end{eqnarray}
and $W_{k_{0}}$ can be expressed as
\begin{equation}
W_{k_{0}}=\sum_{k=1}^{k_{0}}\RR_{\Gammak}\DD_{k-1}.
\label{eq:SumStepError}
\end{equation}
We are now ready to prove that $W=W_{k_0}$ satisfies the four inequalities in \eqref{eq:QRequirement} under our assumptions. The proof uses Lemmas \ref{lem:OP1}-\ref{lem:PTE_infr} in the Appendix.

\textbf{Inequality $(a)$: Bounding $\left\Vert \PP_{\T}W-(UV^{\top}-\llambda\PP_{\T}\EE)\right\Vert _{F}$}

Thanks to \eqref{eq:stepError}, we have the following geometric
convergence
\begin{eqnarray*}
 & & \left\Vert \PP_{\T}W-(UV^{\top}-\llambda\PP_{\T}\EE)\right\Vert _{F}
      = \left\Vert \DD_{k_{0}}\right\Vert _{F}\\
 & = & \left\Vert (\PP_{\T}-\PP_{\T}\RR_{\Gamma^{(k_0)}}\PP_{\T}) \cdots (\PP_{\T}-\PP_{\T}\RR_{\Gamma^{(1)}}\PP_{\T}) \DD_{0} \right\Vert _{F}\\
 & \le & \left(\prod_{k=1}^{k_0}\left\Vert \PP_{\T}-\PP_{\T}\RR_{\Gammak}\PP_{\T}\right\Vert \right)
 \left\Vert UV^{\top}-\llambda \PP_{\T}\EE\right\Vert _{F}\\
 & \overset{(i)}{\le} & e^{-k_{0}}\left(\left\Vert UV^{\top}\right\Vert _{F}+\llambda\left\Vert \PP_{\T}\EE\right\Vert _{F}\right)\\
 & \overset{(ii)}{\le} & n^{-4}\left(n+\llambda n\right)
  \overset{(iii)}{\le}  \frac{\llambda}{\sqrt{n}};
\end{eqnarray*}
here (i) uses Lemma \ref{lem:OP}, (ii)
uses $\left\Vert \PP_{\T}E\right\Vert _{F}\le\left\Vert E\right\Vert _{F}\le n$,
and (iii) is due to our choice of $\llambda$. This proves inequality
(a) in \eqref{eq:QRequirement}.

\textbf{Inequality $(c)$: Bounding $\left\Vert \PP_{\Gamma}W\right\Vert _{\infty}$}

We write
\[
\prod_{i=1}^{k} (\PP_{\T}-\PP_{\T}\RR_{\Gammai}\PP_{\T})
  = (\PP_{\T}-\PP_{\T}\RR_{\Gammak}\PP_{\T}) \cdots (\PP_{\T}-\PP_{\T}\RR_{\Gamma^{(1)}}\PP_{\T})
\]
where the order of multiplication is important. Then we have
\begin{eqnarray}
& & \left\Vert \PP_{\Gamma}W\right\Vert _{\infty} = \left\Vert W_{k_{0}}\right\Vert _{\infty}\nonumber \\
 & \overset{(i)}{\le} & \sum_{k=1}^{k_{0}}\left\Vert \RR_{\Gammak}\DD_{k-1}\right\Vert _{\infty}\nonumber
    \le q^{-1}\sum_{k=1}^{k_{0}}\left\Vert \DD_{k-1}\right\Vert _{\infty}\nonumber \\
 & \overset{(ii)}{=} & q^{-1}\sum_{k=1}^{k_{0}}\left\Vert \prod_{i=1}^{k-1}\left(\PP_{\T}-\PP_{\T}\RR_{\Gammai}\PP_{\T}\right)\DD_{0}\right\Vert _{\infty}\nonumber \\
 & \le &  q^{-1}\sum_{k=1}^{k_{0}}\left\Vert \prod_{i=1}^{k-1}\left(\PP_{\T}-\PP_{\T}\RR_{\Gammai}\PP_{\T}\right)(UV^{\top} -\llambda \PP_\T\PP_{\Omegad}\EE)\right\Vert _{\infty} \nonumber\\
 & & \qquad + q^{-1}\sum_{k=1}^{k_{0}}\left\Vert \prod_{i=1}^{k-1}\left(\PP_{\T}-\PP_{\T}\RR_{\Gammai}\PP_{\T}\right)(-\llambda \PP_{\T}\PP_{\Omegar\backslash\Omegad}\EE)\right\Vert _{\infty};  \nonumber
\end{eqnarray}
here (i) uses \eqref{eq:SumStepError} and (ii) uses \eqref{eq:stepError}. We bound the above two terms separately.

The first term is bounded as
\begin{eqnarray}
 & & q^{-1}\sum_{k=1}^{k_{0}}\left\Vert \prod_{i=1}^{k-1}\left(\PP_{\T}-\PP_{\T}\RR_{\Gammai}\PP_{\T}\right)(UV^{\top}-\llambda \PP_\T\PP_{\Omegad}\EE)\right\Vert _{\infty} \nonumber \\
 & \overset{(i)}{\le} & q^{-1}\sum_{k=1}^{k_{0}}\left(\frac{1}{2}\right)^{k-1}\left\Vert UV^{\top}-\llambda \PP_\T\PP_{\Omegad}\EE\right\Vert _{\infty}\label{eq:sum_inf_norm}\\
 & \le & C\frac{k_{0}^2}{p_{0}(1-\tau)} \left\Vert UV^{\top}-\llambda \PP_\T\PP_{\Omegad}\EE\right\Vert _{\infty}\nonumber \\
 & \overset{(ii)}{\le} & C\frac{k_{0}^2}{p_{0}(1-\tau)} \left(\sqrt{\frac{\mu r}{n^{2}}}+\llambda\alpha\right)\nonumber \\
 & \overset{(iii)}{\le} & \frac{1}{4}\llambda;\label{eq:UVinf}
\end{eqnarray}
Here (i) uses the second part of Lemma \ref{lem:INF} with $\Omega_0=\Gammak$ and $\epsilon_{3}=\frac{1}{4}$, as well as the fact that $\alpha\le\frac{1}{4}$ under the assumptions of Theorem \ref{thm:random},
(ii) uses the incoherence assumptions and Lemma \ref{lem:PTE_infd}, and (iii) holds under the assumptions of Theorem \ref{thm:random}.

For the second term, we can not use the above argument, because $\EE=P_{\Phi}(sgn(S_{0}))$ is not independent of $\Gammai$'s and thus Lemma \ref{lem:INF} does not apply. Instead, we need to utilize the random signs of $\EE := \PP_{\OmegaO}\left(\sgn(A^{\ast})\right)$ (a similar argument appeared in \cite{CanLiMaWri}). Consider the $k$-th term in the sum. We have
\begin{eqnarray}
 &  & q^{-1}\left\Vert \prod_{i=1}^{k-1}(\PP_{\T}-\PP_{\T}\RR_{\Gammai}\PP_{\T})(\llambda \PP_{\T}\PP_{\Omegar\backslash\Omegad}\EE)\right\Vert _{\infty} \nonumber \\
 & = & \llambda q^{-1}\max_{a,b}\left\vert\left\langle e_{a}e_{b}^{\top},\;\prod_{i=1}^{k-1}(\PP_{\T}-\PP_{\T}\RR_{\Gammai}\PP_{\T})(\PP_{\T}\PP_{\Omegar\backslash\Omegad}\EE)\right\rangle\right\vert \nonumber \\
 & = & \llambda q^{-1}\max_{a,b}\left\vert\left\langle \prod_{i=1}^{k-1}(\PP_{\T}-\PP_{\T}\RR_{\Gammaki}\PP_{\T})\left(e_{a}e_{b}^{\top}\right),\; \PP_{\OmegaO\cap(\Omegar\backslash\Omegad)}\left(\sgn(A^*)\right)\right\rangle\right\vert; \nonumber
\end{eqnarray}
here in the last equality we use the self-adjointness of the operators.
Conditioned on $\OmegaO$, $\Omega$, and $\Gammai$'s, $\PP_{\OmegaO\cap(\Omegar\backslash\Omegad)}\left(\sgn(A^*)\right)$ has i.i.d. symmetric
$\pm1$ entries, so Hoeffding's inequality gives,
\begin{eqnarray}
 &  & \mathbb{P}\left(\llambda q^{-1}\left\vert\left\langle \PP_{\OmegaO}\prod_{i=1}^{k-1}(\PP_{\T}-\PP_{\T}\RR_{\Gammaki}\PP_{\T})\left(e_{a}e_{b}^{\top}\right),\; \PP_{\OmegaO\cap(\Omegar\backslash\Omegad)}\left(\sgn(A^*)\right)\right\rangle\right\vert >t \vert \OmegaO, \Omega,\Gammai\textrm{'s}\right) \nonumber\\
 & \le & 2\exp\left(-\frac{2t^{2}}{\left\Vert \llambda q^{-1}\prod_{i=1}^{k-1}(\PP_{\T}-\PP_{\T}\RR_{\Gammaki}\PP_{\T})\left(e_{a}e_{b}^{\top}\right)\right\Vert _{F}^{2}}\right) \nonumber \\
 & \le & 2\exp\left(-\frac{2t^{2}}{\llambda^{2}q^{-2}\left\Vert \prod_{i=1}^{k-1}(\PP_{\T}-\PP_{\T}\RR_{\Gammaki}\PP_{\T})\right\Vert ^{2}\left\Vert \PP_{\T}(e_{a}e_{b})\right\Vert _{F}^{2}}\right) \nonumber\\
 & \le & 2\exp\left(-\frac{t^{2}}{\llambda^{2}q^{-2}\prod_{i=1}^{k-1}\left\Vert \PP_{\T}-\PP_{\T}\RR_{\Gammaki}\PP_{\T}\right\Vert ^{2}\frac{2\mu r}{n}}\right);\label{eq:1}
\end{eqnarray}
here the last inequality uses $\left\Vert \PP_{\T}(e_{a}e_{b}^{\top})\right\Vert _{F}^{2}\le\frac{2\mu r}{n}$, which follows from the incoherence assumptions.
Conditioned on the event $G_k:=\left\{ \left\Vert \PP_{\T}-\PP_{\T}\RR_{\Gammaki}\PP_{\T}\right\Vert \le\frac{1}{2},\; i=1,\ldots k-1\right\} $,
we can integrate out the conditions in (\ref{eq:1}) and obtain
\begin{eqnarray}
 &  & \mathbb{P}\left(\llambda q^{-1}\left\vert\left\langle \prod_{i=1}^{k-1}(\PP_{\T}-\PP_{\T}\RR_{\Gammaki}\PP_{\T})\left(e_{a}e_{b}^{\top}\right),\; \PP_{\OmegaO\cap(\Omegar\backslash\Omegad)}\left(\sgn(A^{\ast})\right)\right\rangle\right\vert >t\vert G_k\right) \nonumber \\
 & \le & 2\exp\left(-\frac{t^{2}}{\llambda^{2}q^{-2}\left(\frac{1}{2}\right)^{k-1}\frac{2\mu r}{n}}\right) \nonumber
\end{eqnarray}
By Lemma \ref{lem:OP}, we know that the event $G_k$ holds with high
probability. Choosing $t=C_{}\left(\frac{1}{2}\right)^{k-1}\frac{\llambda\mu r\log n}{qn}$
with $C_{}$ sufficiently large and using union bound (there is only
polynomially many different $(a,b)$), we conclude that
\begin{equation}
q^{-1} \left\Vert \prod_{i=1}^{k-1}(\PP_{\T}-\PP_{\T}\RR_{\Gammai}\PP_{\T})(\llambda \PP_{\T}\PP_{\Omegar\backslash\Omegad}\EE)\right\Vert _{\infty}\le C_{}\left(\frac{1}{2}\right)^{k-1}\frac{\llambda\mu r\log n}{qn}\le\left(\frac{1}{2}\right)^{k}\cdot\frac{1}{4}\llambda \nonumber
\end{equation}
with high probability; here the second inequality holds because $q\ge C'\frac{\mu r\log n}{n}$ by our choice. Summing over $k$
It follows that
\begin{equation}
\label{eq:Einf}
\sum_{k=1}^{k_{0}}q^{-1}\left\Vert \prod_{i=1}^{k-1}(\PP_{\T}-\PP_{\T}\RR_{\Gammai}\PP_{\T})(\llambda \PP_{\T}\PP_{\Omegar\backslash\Omegad}\EE)\right\Vert _{\infty}\le\frac{1}{4}\llambda.
\end{equation}
Combing \eqref{eq:UVinf} and \eqref{eq:Einf} proves inequality (c) in \eqref{eq:QRequirement}.

\textbf{Inequality $(d)$: Bounding $\left\Vert \PP_{\T^{\bot}} W \right\Vert $}

We have
\begin{eqnarray}
 & & \left\Vert \PP_{\T^{\bot}}W_{k_{0}}\right\Vert
 \overset{(i)}{\le}  \sum_{k=1}^{k_{0}}\left\Vert \PP_{\T^{\bot}}\RR_{\Gammak}\DD_{k-1}\right\Vert \nonumber\\
 & \overset{(ii)}{=} & \sum_{k=1}^{k_{0}}\left\Vert \PP_{\T^{\bot}}\left(\RR_{\Gammak}\DD_{k-1}-\DD_{k-1}\right)\right\Vert \nonumber\\
 & \overset{(iii)}{\le} & \sum_{k=1}^{k_{0}}\left\Vert \left( \RR_{\Gammak} - \mathcal{I}\right) \prod_{i=1}^{k-1}\left(\PP_{\T}-\PP_{\T}\RR_{\Gammai}\PP_{\T}\right)\DD_{0}\right\Vert _{}\nonumber \nonumber\\
 & \le &  \sum_{k=1}^{k_{0}}\left\Vert \left( \RR_{\Gammak} - \mathcal{I}\right) \prod_{i=1}^{k-1}\left(\PP_{\T}-\PP_{\T}\RR_{\Gammai}\PP_{\T}\right)(UV^{\top}-\llambda\PP_\T\PP_{\Omegad}\EE)\right\Vert_{}\nonumber\\
 &&         \qquad + \sum_{k=1}^{k_{0}}\left\Vert \left( \RR_{\Gammak} - \mathcal{I}\right) \prod_{i=1}^{k-1}\left(\PP_{\T}-\PP_{\T}\RR_{\Gammai}\PP_{\T}\right)(-\llambda \PP_{\T}\PP_{\Omegar\backslash\Omegad}\EE)\right\Vert_{};  \label{eq:twoterm2}
\end{eqnarray}
here (i) uses \eqref{eq:SumStepError}, (ii) uses $D_{k}\in \T$, and (iii) uses \eqref{eq:stepError}. We bound the above two terms separately.

The first term is bounded as
\begin{eqnarray}
 &                    & \sum_{k=1}^{k_{0}}\left\Vert \left( \RR_{\Gammak} - \mathcal{I}\right) \prod_{i=1}^{k-1}\left(\PP_{\T}-\PP_{\T}\RR_{\Gammai}\PP_{\T}\right)(UV^{\top}-\llambda\PP_\T\PP_{\Omegad}\EE)\right\Vert_{} \nonumber\\
 & \overset{(i)}{\le} & C\left(\sqrt{\frac{n\log n}{q}}+ d \right)\sum_{k=1}^{k_{0}}\left\Vert \prod_{i=1}^{k-1}\left(\PP_{\T}-\PP_{\T}\RR_{\Gammai}\PP_{\T}\right)(UV^{\top}-\llambda\PP_\T\PP_{\Omegad}\EE) \right\Vert _{\infty}\nonumber\\
 & \overset{(ii)}{\le} & 2C\left(\sqrt{\frac{n\log n}{q}}+ d \right)\left\Vert UV^{\top}-\llambda\PP_\T\PP_{\Omegad}\EE \right\Vert _{\infty} \nonumber\\
 & \overset{(iii)}{\le} & 2C \left(\sqrt{\frac{n\log n}{q}}+ d \right) \left(\sqrt{\frac{\mu r}{n^{2}}} + \llambda\alpha\right)\nonumber\\
 & \overset{(iv)}{\le} & \frac{1}{8}; \label{eq:1sttermbound}
\end{eqnarray}
here (i) uses the second part of Lemma \ref{lem:OP_INF} with $\Omega_0=\Gammak$, (ii) uses \eqref{eq:sum_inf_norm},
(ii) uses the incoherence assumptions and Lemma \ref{lem:PTE_infd}, and (iv) holds under the assumption of Theorem \ref{thm:random}.

For the second term in \eqref{eq:twoterm2}, the above argument fails due to the dependence between $\PP_{\Omegar\backslash\Omegad}\EE$ and $\Gammai$'s. Again we rely on the random signs of $\PP_{\Omegar\backslash\Omegad}\EE = \PP_{\OmegaO\cap(\Omegar\backslash\Omegad)}\sgn(A^*)$, but the situation is more complicated here as we need to use an $\epsilon-$net argument to bound the operator norms.

The key idea is to observe that, though independence does not hold, conditional independence does -- $\Gammai$'s
and $\EE$ are independent conditioned on $\Omega$. This is because $\textrm{supp}(\EE)\subseteq\Omega$ is a random subset of the corrupted entries while $\Gammai\subseteq\Omega^c$ are random subsets of the un-corrupted entries. To isolate this independence, we telescope the operators in the second term in \eqref{eq:twoterm2}. For $k=1,\ldots,k_0$, define the operators
\begin{eqnarray*}
\Ak & = & \PP_{\T}-\PP_{\T}\RR_{\Omegak}\PP_{\T}\\
\Sk & = & \PP_{\T}\RR_{\Omegak}\PP_{\T}-\PP_{\T}\RR_{\Gammak}\PP_{\T}\\
\Bk & = & \RR_{\Omegak}-\mathcal{I}\\
\Tk & = & \RR_{\Gamma_{k}}-\RR_{\Omegak}
\end{eqnarray*}
Observe that $\PP_{\T}-\PP_{\T}\RR_{\Gammak}\PP_{\T}=\Ak+\Sk$, and $\RR_{\Gammak}-\mathcal{I}=\Bk+\Tk$.
The reason for doing so is that, conditioned on $\Omega$, $\Tk$'s
and $\Si$'s are independent of $\EE$. Thus if
a term only involves $\Tk$ and $\Sk$'s (we call it a Type-1 term), it can be bounded in a similar way as the first term in \eqref{eq:twoterm2} using Lemma \ref{lem:OP_INF} and \ref{lem:INF}. For the other terms that involve not only $\Tk$ and $\Sk$'s
but also $\Ai$'s and/or $\Bk$'s (dubbed Type-2 terms),
we bound them using the random signs of $\EE$. (It turns out if one bounds the Type-1 term using the random signs, the resulting
bound is not strong enough, so we need to distinguish these two cases).

Now for the details. Consider the $k$-th term in summands of the second term in \eqref{eq:twoterm2}. Using the above definitions, we have
\begin{eqnarray}
 &  & \left\Vert \left( \RR_{\Gammak} - \mathcal{I}\right) \prod_{i=1}^{k-1}\left(\PP_{\T}-\PP_{\T}\RR_{\Gammai}\PP_{\T}\right)(-\llambda \PP_{\T}\PP_{\Omegar\backslash\Omegad}\EE) \right\Vert \nonumber \\
 & = & \left\Vert \left(\Bk+\Tk\right)\prod_{i=1}^{k-1}(\Ai+\Si)(\llambda \PP_{\T}\PP_{\Omegar\backslash\Omegad}\EE)\right\Vert \label{2ndtermintype}
\end{eqnarray}
We expand the product and sums in the above equation, which results
in a sum of $2^{k}$=poly($n)$ terms since $k\le k_{0}=O(\log n)$.
Among them there is one Type-1 term
\begin{equation}
\Tk\mathcal{S}_{1}\mathcal{S}_{2}\cdots \mathcal{S}_{k-1}(\llambda \PP_{\T}\PP_{\Omegar\backslash\Omegad}\EE),
\end{equation}
and $2^k-1$ Type-2 terms, such as
\begin{eqnarray*}
& \Tk\mathcal{A}_{1}\mathcal{S}_{2} \mathcal{S}_{3}\cdots \mathcal{A}_{k-2} \mathcal{S}_{k-1}(\llambda \PP_{\T}\PP_{\Omegar\backslash\Omegad}\EE),& \\
& \Bk\mathcal{S}_{1}\mathcal{A}_2 \mathcal{S}_{3}\cdots \mathcal{S}_{k-2} \mathcal{A}_{k-1}(\llambda \PP_{\T}\PP_{\Omegar\backslash\Omegad}\EE).&
\end{eqnarray*}

We first bound the Type-1 term. Conditioned on $\Omega$, we have
\begin{eqnarray*}
 &                    & \left\Vert \Tk\mathcal{S}_{1}\mathcal{S}_{2}\cdots \mathcal{S}_{k-1}(\llambda \PP_{\T} \PP_{\Omegar\backslash\Omegad}\EE) \right\Vert\\
 & =                  & \left\Vert \left( \frac{1}{q_1q_2}\PP_{\OmegaOk\cap(\Omegak\cap\Gammad)} - \frac{1}{q_1}\PP_{\Omegak\cap\Gammad}\right) \prod_{i=1}^{k-1} \left(\frac{1}{q_1q_2}\PP_{\T}\PP_{\OmegaOi\cap(\Omegai\cap\Gammad)}\PP_{\T}-\frac{1}{q_1}\PP_{\T}\PP_{\Omegai\cap\Gammad}\PP_{\T}\right)(\llambda \PP_{\T}\PP_{\Omegar\backslash\Omegad} \EE)\right\Vert_{} \\
 & \overset{(i)}{\le} & C \left(\frac{1}{q_1}\sqrt{\frac{n\log n}{q_2}} \right) \left(\frac{1}{2}\right)^{k-1} \left\Vert\llambda \PP_{\T}\PP_{\Omegar\backslash\Omegad}\EE \right\Vert _{\infty}\\
 & \overset{(ii)}{\le} & C' \sqrt{\frac{n\log n}{q_2q_1^2}} \left(\frac{1}{2}\right)^{k} \llambda \sqrt{\frac{\mu r}{n}p_{0} \log n}\\
 & \overset{(iii)}{\le} & \frac{1}{16}\left(\frac{1}{2}\right)^{k}; \label{eq:type1bound}
\end{eqnarray*}
here in (i) we apply the first part of Lemma \ref{lem:OP_INF} with $\Omega_0=\OmegaOk$ and $\Gammao=\Omegak\cap\Gammad$, as well as the first part of Lemma \ref{lem:INF} with $\Omega_0=\OmegaOi$, $\Gammao=\Omegai\cap\Gammad$ and $\epsilon_{3} = \frac{1}{2}q_1$, (ii) uses Lemma \ref{lem:PTE_infr}, and (iii) holds  under the assumption of Theorem \ref{thm:random}.

We next bound the remaining $2^k-1$ Type-2 terms. To this end, we first collect five useful inequalities. Because $\Omegai\sim Ber(q_1)$, the second part of Lemma \ref{lem:OP1} with $\Omega_0 = \Omegai$ and $\epsilon_1=C\frac{\mu r \log n}{nq_1}$ gives that w.h.p.
\begin{eqnarray}
 &     & \left\Vert \Ai \right\Vert = \left\Vert\PP_{\T}-\PP_{\T}\RR_{\Omegai}\PP_{\T} \right\Vert \nonumber\\
 & \le & C\sqrt{\frac{\mu r\log n}{n q_1}}+C\sqrt{\frac{\mu rd}{n}}
   \le   C'\sqrt{\frac{p_{0}(1-\tau)}{\log^3 n}} \label{eq:a}
\end{eqnarray}
The first part of Lemma \ref{lem:OP1} with $\Omega_0=\Omegak$ and $\Gammao=\Gammad$
shows that w.h.p.
\begin{eqnarray}
 &     & \left\Vert \PP_{\T}\Bk\right\Vert =\left\Vert \frac{1}{q_1}\PP_{\T}\PP_{\Omegak\cap\Gammad}-\PP_{\T}\right\Vert \nonumber\\
 & \le & \frac{1}{q_1}\left\Vert \PP_\T\PP_{\Omegak\cap\Gammad} \right\Vert + \left\Vert\PP_\T\right\Vert
    = \frac{1}{q_1}\sqrt{\left\Vert \PP_\T\PP_{\Omegak\cap\Gammad}\PP_\T \right\Vert} + 1 \nonumber\\
 & \le & \frac{1}{q_1}\sqrt{q_1\left\Vert \frac{1}{q_1}\PP_\T\PP_{\Omegak\cap\Gammad}\PP_\T - \PP_\T\PP_{\Gammad}\PP_\T\right\Vert + q_1 \left\Vert\PP_\T\PP_{\Gammad}\PP_\T\right\Vert} + 1 \le  C\sqrt{\frac{1}{q_1}} \le C'\sqrt{\frac{\log n}{1-\tau}}\label{eq:b}
\end{eqnarray}
Similarly, we have w.h.p.
\begin{eqnarray}
 &     & \left\Vert \PP_{\T}\Tk \right\Vert = \left\Vert \frac{1}{q_1q_2}\PP_{\T}\PP_{\OmegaOk\cap\Omegak}-\frac{1}{q_1}\PP_{\T}\PP_{\Omegak}\right\Vert \nonumber\\
 & \le &  \left\Vert \frac{1}{q_1q_2}\PP_{\T}\PP_{\OmegaOk\cap\Omegak} -\PP_\T \right\Vert + \left\Vert \PP_\T-\frac{1}{q_1}\PP_{\T}\PP_{\Omegak}\right\Vert \nonumber\\
 & \le & C\sqrt{\frac{1}{q_1q_2}} + C\sqrt{\frac{1}{q_1}} \le C'\sqrt{\frac{\log^2 n}{p_0(1-\tau)}}. \label{eq:c}
\end{eqnarray}
Applying the first part of Lemma \ref{lem:OP1} twice with (1) $\Omega_0 = \Omegak$, $\Gammao=\Gammad$, $\epsilon_1 = C\sqrt{\frac{\mu \log n }{nq_1}}$ and (2) $\Omega_0 = \OmegaOk\cap\Omegak$, $\Gammao=\Gammad$, $\epsilon_1 = C\sqrt{\frac{\mu \log n }{nq_1q_2}}$ gives w.h.p.
\begin{eqnarray}
 &     & \left\Vert \Sk \right\Vert = \left\Vert \frac{1}{q_1}\PP_\T\PP_{\Omegak\cap\Gammad}\PP_\T - \frac{1}{q_1q_2}\PP_\T\PP_{\OmegaOk\cap(\Omegak\cap\Gammad)}\PP_\T  \right\Vert \nonumber\\
 & \le & \left\Vert \frac{1}{q_1}\PP_\T\PP_{\Omegak\cap\Gammad}\PP_\T - \PP_\T\PP_{\Gammad}\PP_\T\right\Vert + \left\Vert \PP_\T\PP_{\Gammad}\PP_\T- \frac{1}{q_1q_2}\PP_\T\PP_{(\OmegaOk\cap\Omegak)\cap\Gammad}\PP_\T  \right\Vert \nonumber\\
 & \le &  C\sqrt{\frac{\mu r \log n}{nq_1}} + C\sqrt{\frac{\mu r \log n}{nq_1q_2}} \le C'\sqrt{\frac{\mu r \log^3 n}{np_0(1-\tau)}}  \le \frac{1}{4} \label{eq:d}
\end{eqnarray}
Finally, since $\OmegaO\cap(\Omegar\backslash\Omegad)\subseteq\OmegaO\subseteq\OmegaOr$, we apply the first part of Lemma \ref{lem:OP1} with $\Omega_0=\OmegaOr$, $\Gammao=[n]\times[n]$ and $\epsilon_1=\frac{1}{2}$ to obtain w.h.p.
\begin{eqnarray}
\left\Vert \PP_{\OmegaO\cap(\Omegar\backslash\Omegad)}\PP_{\T}\right\Vert  & \le & \left\Vert \PP_{\OmegaOr}\PP_{\T}\right\Vert =\sqrt{\left\Vert \PP_{\T}\PP_{\OmegaOr}\PP_{\T}\right\Vert }.\nonumber\\
 & = & \sqrt{p_{0}\left\Vert \frac{1}{p_{0}}\PP_{\T}\PP_{\OmegaOr}\PP_{\T}-\PP_{\T}+\PP_{\T}\right\Vert }\nonumber\\
 & \le & \sqrt{p_{0}\left\Vert \frac{1}{p_{0}}\PP_{\T}\PP_{\OmegaOr}\PP_{\T}-\PP_{\T}\right\Vert +p_{0}}\nonumber\\
 & \le & \sqrt{2p_{0}} \label{eq:e}
\end{eqnarray}
Now consider one of the Type-2 terms
\begin{equation*}
\XX (\llambda \PP_{\T} \PP_{\Omegar\backslash\Omegad}\EE)\triangleq \Tk \mathcal{S}_1 \mathcal{S}_{2}\cdots \mathcal{S}_{k-2} \mathcal{A}_{k-1}(\llambda \PP_{\T}\PP_{\Omegar\backslash\Omegad}\EE).
\end{equation*}
Let $\XX^{\ast}$ be the adjoint of $\XX$. The last five inequalities \eqref{eq:a}-\eqref{eq:e} yield w.h.p.
\begin{eqnarray}
 \left\Vert \PP_{\OmegaO\cap(\Omegar\backslash\Omegad)}\PP_{\T}\XX^{*}\right\Vert
 & =   & \left\Vert \PP_{\OmegaO\cap(\Omegar\backslash\Omegad)}\PP_{\T} \mathcal{A}_{k-1} \mathcal{S}_{k-2}\cdots \mathcal{S}_{1} \PP_{\T} \Tk \right\Vert \nonumber\\
 & \le & C\sqrt{p_0} \sqrt{\frac{p_0(1-\tau)}{\log^3 n}} \left(\frac{1}{4}\right)^{k-2}  \sqrt{\frac{\log^2 n}{p_0(1-\tau)}} \nonumber\\
 & \le & C'\sqrt{p_0}\left(\frac{1}{4}\right)^{k}.\label{eq:2}
\end{eqnarray}
It is not hard to check that this inequality also holds for the $\XX$'s associated with other Type-2 terms, except for the term $\left(\RR_{\Omegaone-\mathcal{I}}\right)\left(-\llambda\PP_\T \EE\right)$, which is discussed later. We are ready to bound the operator norm of the Type-2 term using a standard $\epsilon$-net argument. Let $\mathbb{S}^{n-1}$ be the unit sphere in $\mathbb{R}^{n}$, and $N$ be an $1/2$-net of $\mathbb{S}^{n-1}$ of size at most $6^{n}$.
The definition and Lipschitz property of the operator norm gives that
\begin{eqnarray*}
 &  & \left\Vert \XX(\llambda \PP_{\T}\PP_{\Omegar\backslash\Omegad}\EE)\right\Vert \\
 & = & \sup_{x,y\in\mathbb{S}^{n-1}}\left\langle xy^{\top},\; \XX(\llambda \PP_{\T}\PP_{\Omegar\backslash\Omegad}\EE)\right\rangle \\
 & \le & 4\sup_{x,y\in N}\left\langle xy^{\top},\; \XX(\llambda \PP_{\T}\PP_{\Omegar\backslash\Omegad}\EE)\right\rangle
\end{eqnarray*}
For a fixed pair $(x,y)\in N\times N$ , we have
\begin{eqnarray*}
 &  & \left\langle xy^{\top},\; \XX(\llambda \PP_{\T}\PP_{\Omegar\backslash\Omegad}\EE)\right\rangle \\
 & = & \llambda\left\langle \PP_{\OmegaO\cap(\Omegar\backslash\Omegad)}\PP_{\T}\XX^{\ast}\left(xy^{\top}\right),\; \sgn(S^{\ast})\right\rangle
\end{eqnarray*}
We condition on the event that (\ref{eq:2}) holds. Because $\sgn(S^{*})$
has i.i.d. symmetric $\pm1$ entries, Hoeffding's inequality gives
\begin{eqnarray*}
 &    & \mathbb{P}\left(\llambda\left\langle \PP_{\OmegaO\cap(\Omegar\backslash\Omegad)}\PP_{\T}\XX^{\ast}\left(xy^{\top}\right),\; \sgn(S^{\ast})\right\rangle \ge\frac{C}{4^{k}}\right)\\
 & \le & 2\exp\left(-\frac{2\cdot\frac{C^2}{4^{2k}}}{\left\Vert \llambda \PP_{\OmegaO\cap(\Omegar\backslash\Omegad)}\PP_{\T}\XX^{\ast}\left(xy^{\top}\right)\right\Vert _{F}^{2}}\right)\\
 & \le & 2\exp\left(-\frac{2\cdot\frac{C^2}{4^{2k}}}{\frac{1}{32^2p_0n(d+1)}\left\Vert \PP_{\OmegaO\cap(\Omegar\backslash\Omegad)}\PP_{\T}\XX^{\ast}\right\Vert ^{2}}\right)\\
 & \le & 2\exp\left(-\frac{C'\cdot\frac{1}{4^{2k}}}{\frac{1}{np_0}\cdot p_0 \frac{1}{4^{2k}}}\right)\\
 & \le & 2\exp\left(-C'n\right)
\end{eqnarray*}
for some constant $C'$ that can be made large. This probability is exponentially small, so we can apply union bound over the $6^{n}$ pairs $(x,y)$ in the
$\epsilon$-net $N\times N$ and conclude that w.h.p.
\begin{equation*}
\left\Vert X(\llambda \PP_{\T} \EE)\right\Vert \le\frac{C}{4^{k}} = C\frac{1}{2^k}\frac{1}{2^k}
\end{equation*}
For the exceptional term $\left(\RR_{\Omegaone-\mathcal{I}}\right)\left(-\llambda\PP_\T \EE\right)$, a similar bound holds as follows. The proof can be found in the Appendix.
\begin{lemma}
\label{lem:exceptionterm}
Under the assumption of Theorem \ref{thm:random}, the following holds with high probability
\[
\left\Vert\left(\RR_{\Omegaone-\mathcal{I}}\right)\left(-\llambda\PP_\T \EE\right)\right\Vert \le \frac{1}{32}.
\]
\end{lemma}
Summing over all $2^{k}-1=\textrm{poly}(n)$ Type-2 terms and combining with the bound \eqref{eq:type1bound} for the Type-1 term, it follows that the right hand side of \eqref{2ndtermintype} is bounded by $\frac{1}{8\cdot2^k}$. Summing over $k=1,2,\ldots,k_0$ bounds the second term in \eqref{eq:twoterm2} by $\frac{1}{8}$, which, together with the bound \eqref{eq:1sttermbound} for the first term, completes the proof of inequality (d) in  \eqref{eq:QRequirement}.

\textbf{Inequality $(e)$: Bounding $\left\Vert \PP_{\T^{\bot}} \llambda \EE \right\Vert $}

A standard argument about the norm of a matrix
with i.i.d. entries \cite{Versh} and \cite[Proposition 3]{cspw} give
\begin{eqnarray*}
\left\Vert \PP_{\T^{\bot}} \llambda \EE \right\Vert \le \llambda \left(\left\Vert \PP_{\Omegar\backslash\Omegad}\EE\right\Vert + \left\Vert\PP_{\Omegad}\EE\right\Vert\right)   \le  \frac{1}{32\sqrt{p_{0}(d+1)n\log n}}\cdot\left(4\sqrt{np_{0}\tau}+d\right).
\end{eqnarray*}
Under the assumption of Theorem \ref{thm:random}, the right hand side is no larger than $\frac{1}{4}$. Therefore, inequality (e) in \eqref{eq:QRequirement} holds.

This completes the proof of Theorem \ref{thm:random}. As mentioned in section \ref{sub:derand}, Theorem \ref{thm:randomfixed} also follows. 

\section{Proof of Theorem~\ref{theorem:mainresult}}
\label{sec:worst_case_proof}
\noindent The proof is along the lines of that in \cite{cspw} and has three steps: {\em (a)} writing down a sufficient optimality condition, stated in terms of a dual certificate, for $(\PP_{\OmegaO}(A^{\ast}),\; B^{\ast})$ to be the optimum of the convex program \eqref{eq:opt_problem}, {\em (b)} constructing a particular candidate dual certificate, and, {\em (c)} showing that under the imposed conditions this candidate does indeed certify that $(\PP_{\OmegaO}(A^{\ast}),\; B^{\ast})$ is the optimum. Part {\em (b)} is the "art" in this method; different ways to devise dual certificates can yield different sufficient conditions for exact recovery. Indeed this is the main difference between this paper and \cite{cspw}.

\subsubsection{Optimality conditions}

%

For the sake of completeness, we restate here a first-order sufficient condition that guarantees $(\PP_{\OmegaO}(A^{\ast}),\; B^{\ast})$ to be the optimum of \eqref{eq:opt_problem}. The reader is referred to \cite{cspw} for a proof.

\begin{lemma}[{\bf A Sufficient Optimality Condition} \cite{cspw}]
The pair $(\PP_{\OmegaO}(A^{\ast}),\; B^{\ast})$ is the unique optimal solution of \eqref{eq:opt_problem} if
\begin{itemize}
\item [(a)] $\Gamma^c\cap\mathcal{T}=\{\mathbf{0}\}$.
\item [(b)] There exists a dual matrix $Q\in\mathbb{R}^{n_1\times n_2}$ satisfying $\PP_{\OmegaO^c}(Q)=0$ and
\begin{equation}
\begin{aligned}
&\PP_{\mathcal{T}}(Q)=UV^\top&\qquad &\|\PP_{\mathcal{T}^{\perp}}(Q)\|< 1\\
&\PP_{\Gamma^c}(Q)=\gamma \PP_{\OmegaO}(\sgn(A^*))&\qquad & \|\PP_{\Gamma}(Q)\|_{\infty}<\gamma.
\end{aligned}
\label{eq:orig_dual_cond}
\end{equation}
\end{itemize}
\label{lem:detopt}
\end{lemma}
Lemma \ref{lem:detopt} provides a first-order sufficient condition for $(\PP_{\OmegaO}(A^*), B^*)$ to be the optimum of \eqref{eq:opt_problem}. Condition (a) in the lemma guarantees that the sparse matrices and low-rank matrices can be distinguished without ambiguity. In other words, any given matrix can not be both sparse and low-rank except the zero matrix. The following lemma gives a sufficient guarantee for the condition (a). We construct the dual matrix $Q$ in the next subsection and prove condition (b) afterwards.\\

\begin{lemma}
If $\alpha<1$, then $\Gamma^c\cap\mathcal{T}=\{\mathbf{0}\}$.
\label{lemma:zero-intersection}
\end{lemma}
\begin{proof}
It is clear that $\{\mathbf{0}\}\in\Gamma^c\cap\mathcal{T}$. In order to obtain a contradiction assume that there exists a non-zero matrix $M\in\Gamma^c\cap\mathcal{T}$. By idempotency of orthogonal projections, we have $M=\PP_{\Gamma^c}(M)=\PP_{\mathcal{T}}\left(\PP_{\Gamma^c}(M)\right)$ and hence
\begin{equation}
\begin{aligned}
&\left\|\PP_{\mathcal{T}}\left(\PP_{\Gamma^c}(M)\right)\right\|_{\infty}\\
&\quad=\left\|UU^\top\PP_{\Gamma^c}(M)+\PP_{\Gamma^c}(M)VV^\top-UU^\top\PP_{\Gamma^c}(M)VV^\top\right\|_{\infty}\\
&\quad\leq\left\|UU^\top\PP_{\Gamma^c}(M)\right\|_{\infty}+\left\|\PP_{\Gamma^c}(M)VV^\top\right\|_{\infty} +\left\|UU^\top\PP_{\Gamma^c}(M)VV^\top\right\|_{\infty}\\
&\quad\leq\max_{i}\left\|UU^\top\mathbf{e}_i\right\|\max_{j}\left\|\PP_{\Gamma^c}(M)\mathbf{e}_j\right\| +\max_{j}\left\|\mathbf{e}_j^\top\PP_{\Gamma^c}(M)\right\|\max_{i}\left\|VV^\top\mathbf{e}_i\right\|\\
&\quad\qquad\qquad+\max_{j}\left\|\mathbf{e}_j^\top UU^\top\right\|\left\|\PP_{\Gamma^c}(M)\right\|\max_{i}\left\|VV^\top\mathbf{e}_i\right\|\\
&\quad\leq\max_{i}\left\|UU^\top\mathbf{e}_i\right\|\sqrt{d}\left\|\PP_{\Gamma^c}(M)\right\|_\infty +\sqrt{d}\left\|\PP_{\Gamma^c}(M)\right\|_\infty\max_{i}\left\|VV^\top\mathbf{e}_i\right\|\\
&\quad\qquad\qquad+\max_{i}\left\|UU^\top\mathbf{e}_i\right\|d\left\|\PP_{\Gamma^c}(M)\right\|_\infty\max_{i}\left\|VV^\top\mathbf{e}_i\right\|\\
&\quad\leq\alpha\|\PP_{\Gamma^c}(M)\|_{\infty}=\alpha\|\PP_{\mathcal{T}}\left(\PP_{\Gamma^c}(M)\right)\|_{\infty}.
\end{aligned}
\label{eq:upper_bound}
\end{equation}
Here, we used the fact that $\sqrt{\frac{\mu r d}{n_1}}\sqrt{\frac{\mu r d}{n_2}}\leq \sqrt{\frac{\mu r d}{\max(n_1,n_2)}}$ since both terms do not exceed $1$ by assumption. Hence, $\|M\|_{\infty}=\mathbf{0}$ or equivalently, $M=\mathbf{0}$. This is a contradiction.\\
\end{proof}

\subsubsection{Dual Certificate}

\noindent We now describe our main innovation, a new way to construct the candidate dual certificate $Q$, which is different from the ones in \cite{cspw}. We construct $Q$ as the minimum norm solution to the equality constraints in Lemma \ref{lem:detopt}. As a first step, consider two matrices $Q_a$ and $Q_b$ defined as follows: with $M^*=\gamma\sgn(A^*)$ and $N^*=UV^*$, let
\begin{equation}
\begin{aligned}
Q_a&=\!M^*\!\!\!-\!\PP_{\mathcal{T}}(M^*)\!+\!\PP_{\Gamma^c}\left(\PP_{\mathcal{T}}(M^*)\right)\!-\!\PP_{\mathcal{T}}\left(\PP_{\Gamma^c}\left(\PP_{\mathcal{T}}(M^*)\right)\right)\!+\!\cdot\!\cdot\!\cdot\\
Q_b&=\!N^*\!\!\!-\!\PP_{\Gamma^c}(N^*)\!+\!\PP_{\mathcal{T}}\left(\PP_{\Gamma^c}(N^*)\right)\!-\!\PP_{\Gamma^c}\left(\PP_{\mathcal{T}}\left(\PP_{\Gamma^c}(N^*)\right)\right)\!+\!\cdot\!\cdot\!\cdot
\end{aligned}
\nonumber
\end{equation}
Lemma \ref{lem:qaqb} below establishes that $Q_a$ and $Q_b$ as described above are well-defined, i.e., it establishes that  the infinite summations converge, under the conditions of the theorem. Note that when this is the case, we have that
\begin{equation}
\begin{aligned}
&\PP_{\mathcal{T}}\left(Q_b\right)=UV^\top\qquad &\qquad \PP_{\mathcal{T}}\left(Q_a\right)=\mathbf{0}\\
&\PP_{\Gamma^c}\left(Q_a\right)=\gamma\PP_{\OmegaO}(\sgn(A^*))\qquad &\qquad \PP_{\Gamma^c}\left(Q_b\right)=\mathbf{0}.
\end{aligned}
\label{eq:equality_dual_cond}
\end{equation}
From \eqref{eq:equality_dual_cond}, it is clear that $Q=Q_a+Q_b$ satisfies the equality conditions in \eqref{eq:orig_dual_cond} and also $\PP_{\OmegaO^c}(Q)=0$. In the next subsection, we will show that the inequality conditions are also satisfied under the assumptions of the theorem~\ref{theorem:mainresult}.

\begin{lemma} \label{lem:qaqb}
If $\alpha<1$, then $Q_a$ and $Q_b$ exist, i.e., the sums converge.
\end{lemma}
\begin{proof}
For any matrix $W\in\mathbb{R}^{n_1\times n_2}$, let $\mathbf{S}_W=W+\PP_{\mathcal{T}}\left(\PP_{\Gamma^c}(W)\right) +\PP_{\mathcal{T}}\left(\PP_{\Gamma^c}\left(\PP_{\mathcal{T}}\left(\PP_{\Gamma^c}(W)\right)\right)\right)+\cdot\cdot\cdot$. It suffices to show that $\mathbf{S}_W$ converges for all $W$ since $Q_a=M^*-\PP_{\Gamma}\left(\mathbf{S}_{\PP_{\mathcal{T}}(M^*)}\right)$ and $Q_b=\mathbf{S}_{N^*-\PP_{\Gamma^c}(N^*)}$. Notice that
$\|\PP_{\mathcal{T}}\left(\PP_{\Gamma^c}\left(W \right)\right)\|_{\infty} \leq\alpha\|\PP_{\Gamma^c}(W)\|_{\infty}\leq\alpha\|W\|_{\infty}$ as shown in \eqref{eq:upper_bound} and hence $\mathbf{S}_W$ geometrically converges.\\
\end{proof}

\subsubsection{Certification}

\noindent Considering $Q=Q_a+Q_b$ as a candidate for dual matrix, we need to show the conditions in \eqref{eq:orig_dual_cond} are satisfied under the conditions of the theorem. As we showed in the previous subsection, the equality conditions are satisfied by construction of $Q_a$ and $Q_b$. To prove the inequality conditions, we first bound the projection of $Q$ into orthogonal complement spaces in next lemma.

\begin{lemma}
If $\alpha<1$, then
\begin{equation}
\begin{aligned}
\|\PP_{\Gamma}(Q)\|_{\infty}&\leq\frac{1}{1-\alpha}\left(\sqrt{\frac{\mu r}{n_1 n_2}}+\alpha\gamma\right)\\
\|\PP_{\mathcal{T}^{\perp}}(Q)\|&\leq\frac{\eta d}{1-\alpha}\left(\sqrt{\frac{\mu r}{n_1 n_2}}+\gamma\right).
\end{aligned}
\nonumber
\end{equation}
\label{lemma:perpbounds}
\end{lemma}
\begin{proof}
Using the definition of $\mathbf{S}_W$ for any matrix $W\in\mathbb{R}^{n_1\times n_2}$, we get $\left\|\mathbf{S}_W\right\|_{\infty}\leq\frac{1}{1-\alpha}\left\|W\right\|_{\infty}$, because of the geometrical convergence. Thus, we have
\begin{equation}
\begin{aligned}
\|\PP_{\Gamma}(Q)\|_{\infty} &=\|\PP_{\Gamma}\left(\mathbf{S}_{N^*-\PP_{\mathcal{T}}(M^*)}\right)\|_{\infty}\\
&\leq\|\mathbf{S}_{N^*-\PP_{\mathcal{T}}(M^*)}\|_{\infty}\\
&\leq\frac{1}{1-\alpha}\|N^*-\PP_{\mathcal{T}}(M^*)\|_{\infty}\\ &\leq\frac{1}{1-\alpha}\left(\|N^*\|_{\infty}+\|\PP_{\mathcal{T}}(M^*)\|_{\infty}\right)\\
&\leq\frac{1}{1-\alpha}\left(\|N^*\|_{\infty}+\alpha\|M^*\|_{\infty}\right)\\ &\leq\frac{1}{1-\alpha}\left(\sqrt{\frac{\mu r}{n_1 n_2}}+\alpha\gamma\right).
\end{aligned}
\nonumber
\end{equation}
In the last inequality we use the incoherence assumptions for sparse and low-rank matrix. By orthonormality of $U$ and $V$, we have $\|\mathbf{I}-UU^\top\|\leq 1$ and $\|\mathbf{I}-VV^\top\|\leq 1$. Hence,
\begin{equation}
\begin{aligned}
&\|\PP_{\mathcal{T}^{\perp}}(Q)\|\\ &=\|\PP_{\mathcal{T}^{\perp}}\left(M^*-\PP_{\Gamma^c}\left(\mathbf{S}_{N^*-\PP_{\mathcal{T}}(M^*)}\right)\right)\|\\ &=\|\left(\mathbf{I}-UU^\top\right)\left(M^*-\PP_{\Gamma^c}\left(\mathbf{S}_{N^*-\PP_{\mathcal{T}}(M^*)}\right)\right)\left(\mathbf{I}-VV^\top\right)\|\\
&\leq\|M^*-\PP_{\Gamma^c}\left(\mathbf{S}_{N^*-\PP_{\mathcal{T}}(M^*)}\right)\|\\ &\leq\eta d \|M^*-\PP_{\Gamma^c}\left(\mathbf{S}_{N^*-\PP_{\mathcal{T}}(M^*)}\right)\|_{\infty}\\
&\leq \eta d\left(\|M^*\|_{\infty}+\|\mathbf{S}_{N^*-\PP_{\mathcal{T}}(M^*)}\|_{\infty}\right)\\ &\leq \eta d\left(\gamma+\frac{1}{1-\alpha}\left(\sqrt{\frac{\mu r}{n_1 n_2}}+\alpha\gamma\right)\right)\\
&\leq\frac{\eta d}{1-\alpha}\left(\sqrt{\frac{\mu r}{n_1 n_2}}+\gamma\right).
\end{aligned}
\nonumber
\end{equation}
Here, again we are using the incoherence assumptions on the sparse and low-rank matrix. This concludes the proof of the lemma.\\
\end{proof}

\noindent Finally to satisfy \eqref{eq:orig_dual_cond}, we require
\begin{equation}
\begin{aligned}
&\|\PP_{\mathcal{T}^{\perp}}(Q)\|&\leq&\quad\frac{\eta d}{1-\alpha}\left(\sqrt{\frac{\mu r}{n_1 n_2}}+\gamma\right)&<&1\\
&\|\PP_{\Gamma}(Q)\|_{\infty}&\leq&\quad\frac{1}{1-\alpha}\left(\sqrt{\frac{\mu r}{n_1 n_2}}+\alpha\gamma\right)&<&\gamma\\
\end{aligned}\quad .
\nonumber
\end{equation}
Combining these two inequalities, we get
\begin{equation}
\frac{1}{1-2\alpha}\sqrt{\frac{\mu r}{n_1n_2}}<\gamma<\frac{1-\alpha}{\eta d}-\sqrt{\frac{\mu r}{n_1n_2}}
\nonumber
\end{equation}
as stated in the assumptions of the theorem.

\section{Experiments}
\noindent In this section, we illustrate the power of our method via some simulation results. These results show that the behavior of the algorithm agrees with the theoretical results.\\

\noindent We investigate how the algorithm performs as the size of the low-rank matrix gets larger. In other words, we try to see how the requirements for the success of our algorithm change as the size of the matrix grows. These simulation results show that the conditions get relaxed more and more as $n$ increases. We run three experiments as follows:\\

\begin{itemize}
\item [(1)] {\bf Minimum Required Observation Probability:} We generate a rank two matrix ($r=2$) of size $n$ by multiplying a random $n\times2$ matrix and a random $2\times n$ matrix, and then corrupt the entries randomly with probability $\tau = 0.1$ without any adversarial noise ($d=0$). The entries of the corrupted matrix are observed independently with probability $p_0$. We then solve \eqref{eq:opt_problem} using the method in \cite{LinALM}. Success is declared if we recover the low-rank matrix with a relative error less than $10^{-6}$ measured in Frobenius norm. The experiment is repeated $10$ times and we count the frequency of success. For any fixed number $n$, if we start from $p_0=1$ and decrease $p_0$, at some point, the frequency of success jumps from one to zero, i.e., we observe a phase transition. In Fig.~\ref{fig:rho}, we plot the $p_0$ at which the phase transition happens versus the size of the matrix. This experiment shows that the phase transition $p_0$ goes to zero as $n$ increases as predicted by the theorem.\\

\item [(2)] {\bf Maximum Tolerable Corruption Probability:} Similarly as before, we generate a rank two matrix ($r=2$) of size $n$, with observation probability $p_0=0.9$ and without any adversarial noise ($d=0$). For any fixed number $n$, if we start from $\tau=0$ and increase $\tau$, at some point, the frequency of success jumps from one to zero. Fig.~\ref{fig:tau} illustrates how the phase transition $\tau$ changes as the size of the matrix increases. This experiment shows that higher probability of corruptions can be tolerated as the size of the matrix increases as predicted by the theorem.\\

\item [(3)] {\bf Maximum Tolerable Adversarial/Deterministic Noise:} Similarly as before, we generate a rank two matrix ($r=2$), of size $n$, with observation probability $p_0=0.5$ and corruption probability $\tau=0.1$. We add the adversarial noise in the form of a $d \times d$ block of $1$'s lying on the diagonal of the original matrix. Notice that potentially it is a hard case to recover the low-rank matrix since all the adversarial corruptions are burst as oppose to be spread over the matrix (Bernoulli corruptions). We find the maximum possible $d$ such that the frequency of success to goes from $1$ to $0$ (phase transition). In Fig.~\ref{fig:d}, we plot this phase transition $d$ versus the size of the matrix and as the deterministic theorem predicts, it grows linearly in $n$.\\

\end{itemize}

\begin{figure}[!t]
\centering
\includegraphics[width=3in]{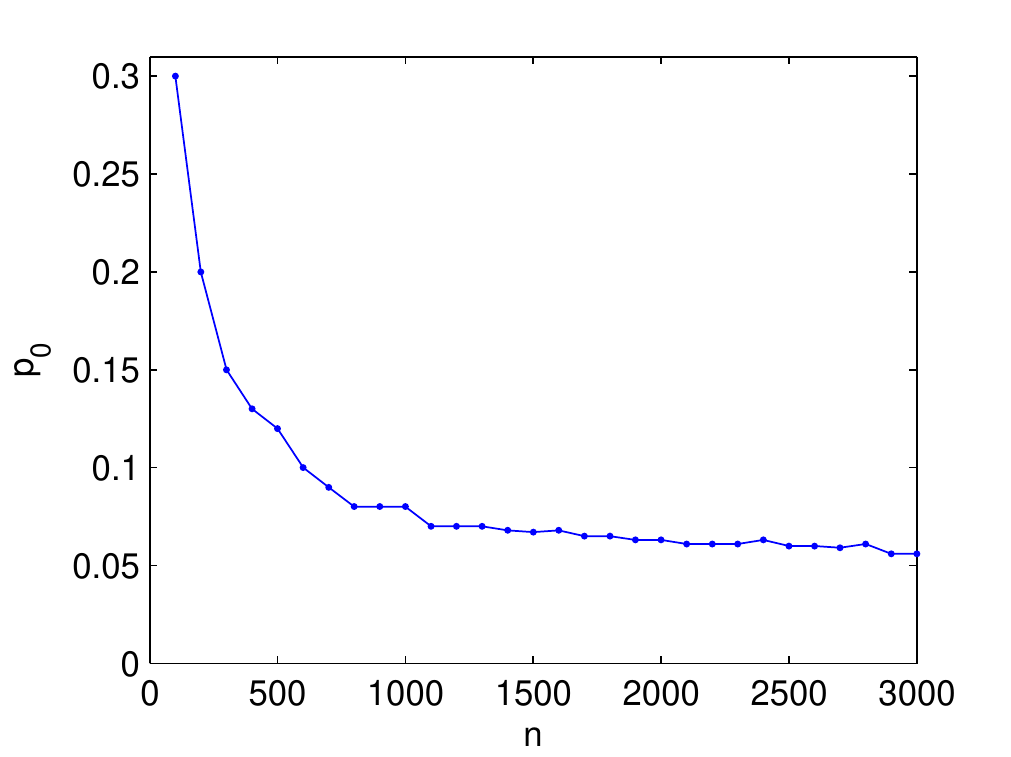}
\caption{For a rank two matrix of size $n$, with probability of corruption $\tau=0.1$ and no adversarial noise ($d=0$), we plot the minimum probability of observation $p_0$ required for successful recovery of the low-rank matrix as $n$ gets larger.}
\label{fig:rho}
\end{figure}

\begin{figure}[!t]
\centering
\includegraphics[width=3in]{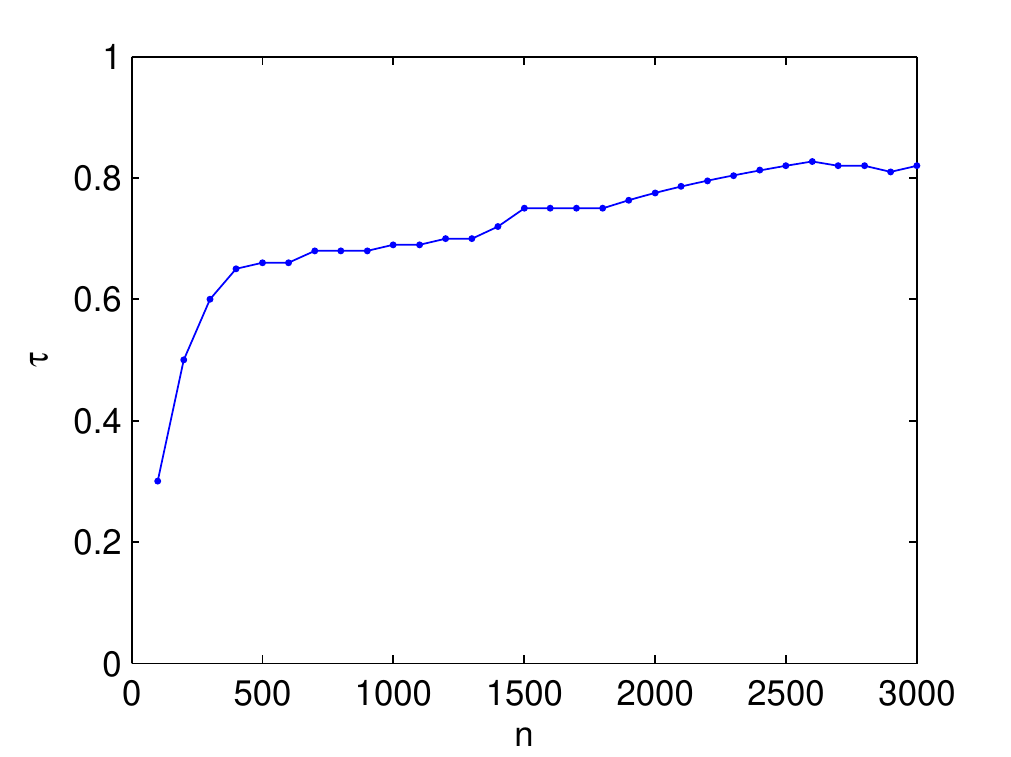}
\caption{For a rank two matrix of size $n$, with probability of observation $p_0=0.9$ and no adversarial noise ($d=0$), we plot the maximum probability of corruption $\tau$ tolerable for successful recovery of the low-rank matrix as $n$ gets larger.}
\label{fig:tau}
\end{figure}

\begin{figure}[!t]
\centering
\includegraphics[width=3in]{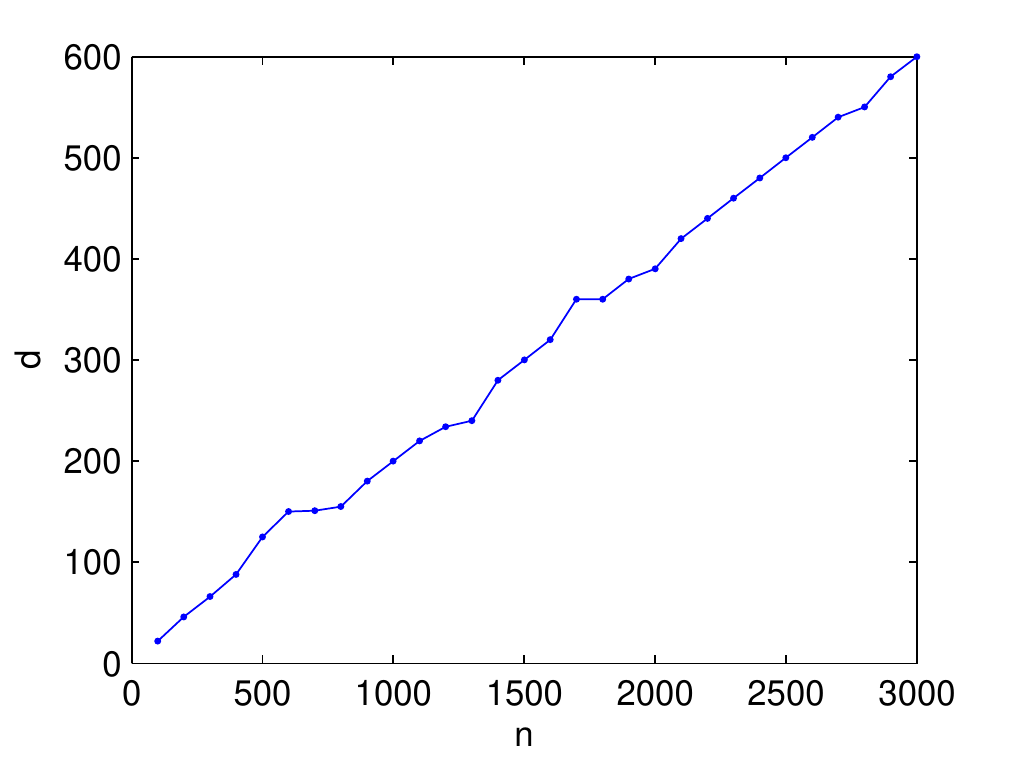}
\caption{For a rank two matrix of size $n$, with probability of observation $p_0=0.5$ and probability of corruption $\tau=0.1$, and with adversarial/deterministic noise in the form of a $d\times d$ block of $1$'s lying on the diagonal of the matrix, we plot the maximum size of the adversarial noise $d$ tolerable for successful recovery of the low-rank matrix as $n$ gets larger.}
\label{fig:d}
\end{figure}





\bibliographystyle{IEEEtran}
\bibliography{IEEEabrv,MSBIB}

\newpage

\appendix


Here we provide several technical lemmas that is needed in the proof of the unified guarantees. We first state the non-commutative Bernstein inequality, which is useful in the sequel. The version presented below is first proved in \cite{Recht,Gro}
and later sharpened in \cite{Tropp}.
\begin{lemma}
\label{lem:bernstein}\cite[Remark 6.3]{Tropp} Consider
a finite sequence $\{Z_{k}\}$ of independent, random $n_{1}\times n_{2}$
matrices that satisfy the assumption $\mathbb{E}Z_{k}=0$ and $\left\Vert Z_{k}\right\Vert \le D$
almost surely. Let $\sigma^{2}=\max\left\{ \left\Vert \sum_{k}\mathbb{E}\left[Z_{k}Z_{k}^{\top}\right]\right\Vert ,\;\left\Vert \sum_{k}\mathbb{E}\left[Z_{k}^{\top}Z_{k}\right]\right\Vert \right\} $.
Then for all $t\ge0$ we have
\begin{eqnarray}
 \mathbb{P}\left[\left\Vert \sum Z_{k}\right\Vert \ge t\right]
 & \le & (n_1+n_2)\exp\left( -\frac{t^2}{2\sigma^2 + \frac{2}{3}Dt} \right) \\ \label{eq:bernstein0}
 & \le &
     \begin{cases}
         (n_{1}+n_{2})\exp\left(-\frac{3t^{2}}{8\sigma^{2}}\right),
 \quad\textrm{for }t\le\frac{\sigma^{2}}{D}.\\
 (n_{1}+n_{2})\exp\left(-\frac{3t^{}}{8D}\right),
 \quad\textrm{for }t\ge\frac{\sigma^{2}}{D}.
     \end{cases}
 \label{eq:bernstein}
\end{eqnarray}
\end{lemma}

W.L.O.G. we only consider the case $n_1=n_2=n$. Recall that we have defined $\alpha=\sqrt{\frac{\mu r d}{n_1}}+\sqrt{\frac{\mu r d}{n_2}} + \sqrt{\frac{\mu r d}{\max\{n_1,n_2\}}}=3\sqrt{\frac{\mu r d}{n}}$. Under the assumptions of Theorem \ref{thm:random}, $\alpha$ is a sufficiently small constant bounded away from $1$. We will make use of the following estimates $\left\Vert \PP_{\T}(e_{i}e_{j}^{\top})\right\Vert _{F}^{2} \le \frac{2\mu r}{n},\;\forall i,j$,  which follow from the incoherence assumptions of $U$ and $V$.

We start with the proof of Lemma \ref{lem:OP}. We need one simple lemma for the deterministic set $\Gammad^{c}$.
\begin{lemma}
\label{lem:det_Fro_bound}
For any matrix $Z\in \T$,
we have
\[
\left\Vert \PP_{\Gammad^{c}}(Z)\right\Vert _{F}\le \alpha\left\Vert Z\right\Vert _{F}
\]
\end{lemma}
\begin{proof}
Since $Z \in \T$, $Z=UX^{\top}+U^{\bot}YV^{\top}$ for some $X,Y\in\mathbb{R}^{n\times r}$. For $1\le j\le n$, incoherence of $B^*$ gives
\begin{eqnarray*}
\left\Vert UX^{\top}e_{j}\right\Vert _{\infty}  = \max_{i}\left|e_{i}^{\top}UX^{\top}e_{j}\right| \le  \sqrt{\frac{\mu r}{n}}\left\Vert X^{\top}e{}_{j}\right\Vert _{2}.
\end{eqnarray*}
Therefore, we have
\begin{eqnarray*}
\left\Vert \PP_{\Gammad^{c}}(UX^{\top})e_{j}\right\Vert _{2}\le\sqrt{d}\left\Vert UX^{\top}e_{j}\right\Vert _{\infty}\le\alpha\left\Vert X^{\top}e_{j}\right\Vert _{2}.
\end{eqnarray*}
It follows that
\begin{eqnarray*}
 &  & \left\Vert \PP_{\Gammad^{c}}(UX^{\top})\right\Vert _{F}^{2}
  =  \sum_{j}\left\Vert \PP_{\Gammad^{c}}(UX^{\top})e_{j}\right\Vert _{2}^{2}\\
 & \le & \sum_{j}\alpha^{2}\left\Vert X^{\top}e_{j}\right\Vert _{2}^{2}
  =  \alpha^{2}\left\Vert X^{\top}\right\Vert _{F}^{2}
 \end{eqnarray*}
Similarly, we have $\left\Vert \PP_{\Gammad^{c}}(U^{\bot}YV^{\top})\right\Vert _{F}^{2}\le\alpha^{2}\left\Vert Y \right\Vert _{F}^{2}$.
The lemma then follows from the triangular inequality and $\left\Vert Z \right\Vert_{F}^2 = \left\Vert X \right\Vert_{F}^2 + \left\Vert Y \right\Vert_{F}^2$.
\end{proof}

We now turn to the proof of Lemma \ref{lem:OP}. In fact, we will prove a slightly more general result as below.

\begin{lemma}
\label{lem:OP1}
Suppose $\Omega_{0}$ is a set of indices obeying $\Omega_{0}\sim$Ber$(p)$, and $\Gammao$ is a fixed set of indices.
\begin{enumerate}
  \item For any $\beta>1$, we have
    \[
    \left\Vert p^{-1}\PP_{\T}\PP_{\Omega_{0}\cap\Gammao}\PP_{\T} - \PP_{\T}\PP_{\Gammao}\PP_{\T} \right\Vert \le\epsilon_{1}
    \]
    with probability at least $1-2n^{2-2\beta}$
    provided $1>\epsilon_1 \ge \sqrt{\frac{32\beta\mu r \log n }{3np}}$.
  \item If in addition, $\Gammao=\Gammad$, where $\Gammad$ satisfies the assumptions in Theorem \ref{thm:random}, then
    \[
    \left\Vert p^{-1}\PP_{\T}\PP_{\Omega_{0}\cap\Gammad}\PP_{\T} - \PP_{\T} \right\Vert \le\epsilon_{1} + \alpha
    \]
    with the same probability.
\end{enumerate}
\end{lemma}


\begin{proof}
We will use Lemma \ref{lem:bernstein} to bound the operator norm of the random component $p^{-1}\PP_{\T}\PP_{\Omega_{0}\cap\Gammao}\PP_{\T}-\PP_{\T}\PP_{\Gammao}\PP_{\T}$. To this end, we need to write the random component as a sum of zero-mean, independent random variables, and then show that each of them is bounded almost surely and their sum has small second moment. Now for the details. For $(i,j)\in\Gammao$, define the indicator random variables $\delta_{ij}=\mathbf{1}_{\{(i,j)\in\Omega_{0}\cap \Gammao\}}$;
so $\delta_{ij}$ equals one with probability $p$ and zero otherwise, and is independent of all others.
For any $Z\in \T$, observe that $Z_{i,j}=\left\langle e_{i}e_{j}^{\top},\; Z\right\rangle $
for $(i,j)\in\Gammao$, and thus
\begin{eqnarray*}
 &  & p^{-1}\PP_{\T}\PP_{\Omega_{0}\cap\Gammao}\PP_{\T}Z-\PP_{\T}\PP_{\Gammao}\PP_{\T}Z\\
 & = & \sum_{(i,j)\in\Gammao}\left(p^{-1}\delta_{ij}-1\right)\left\langle e_{i}e_{j}^{\top},\; Z\right\rangle \PP_{\T}(e_{i}e_{j})^{\top}\\
 & \triangleq & \sum_{(i,j)\in\Gammao}\mathcal{S}_{ij}(Z).
\end{eqnarray*}
Here $\mathcal{S}_{ij}:\;\mathbb{R}^{n\times n}\mapsto\mathbb{R}^{n\times n}$
is a self-adjoint random operator with $\mathbb{E}\left[\mathcal{S}_{ij}\right]=0$.
To use the non-commutative Bernstein inequality, we need to bound $\left\Vert \mathcal{S}_{ij}\right\Vert $,
and $\left\Vert \mathbb{E}\left[\sum_{(i,j)\in\Gammao}\mathcal{S}_{ij}{}^{2}\right]\right\Vert $.
To this end, we have
\begin{eqnarray*}
\left\Vert \mathcal{S}_{ij}\right\Vert  & = & \sup_{\left\Vert Z\right\Vert _{F}=1}\left\Vert \left(p^{-1}\delta_{ij}-1\right)\left\langle \PP_{\T}(e_{i}e_{j}^{\top}),\; Z\right\rangle \PP_{\T}(e_{i}e_{j})^{\top}\right\Vert \\
 & \le & \sup_{\left\Vert Z\right\Vert _{F}=1}p^{-1}\left\Vert \PP_{\T}(e_{i}e_{j}^{\top})\right\Vert _{F}^{2}\left\Vert Z\right\Vert _{F}
 \le \frac{2\mu r}{np}
\end{eqnarray*}
On the other hand, for any $Z\in \T$ we have $\mathcal{S}_{ij}^{2}(Z)=\left(p^{-1}\delta_{ij}-1\right)^{2}\left\langle Z_{i,j}\PP_{\T}(e_{i}e_{j})^{\top},\; e_{i}e_{j}^{\top}\right\rangle \PP_{\T}(e_{i}e_{j}^{\top})$.
Therefore
\begin{eqnarray*}
 &  & \left\Vert \mathbb{E}\left[\sum_{(i,j)\in\Gammao}\mathcal{S}_{ij}^{2}(Z)\right]\right\Vert _{F}\\
 & = & \left(p^{-1}-1\right)\left\Vert \sum_{(i,j)\in\Gammao}\left\Vert \PP_{\T}(e_{i}e_{j})^{\top}\right\Vert _{F}^{2}Z_{i,j}\PP_{\T}(e_{i}e_{j}^{\top})\right\Vert _{F}\\
 & \le & \left(p^{-1}-1\right)\left\Vert \sum_{(i,j)\in\Gammao}\left\Vert \PP_{\T}(e_{i}e_{j})^{\top}\right\Vert _{F}^{2}Z_{i,j}(e_{i}e_{j}^{\top})\right\Vert _{F}\\
 & \le & \left(p^{-1}-1\right)\frac{2\mu r}{n}\left\Vert \sum_{(i,j)\in\Gammao}Z_{i,j}(e_{i}e_{j}^{\top})\right\Vert _{F}\\
 & = & \left(p^{-1}-1\right)\frac{2\mu r}{n}\left\Vert \PP_{\Gammao}(Z)\right\Vert _{F}
  \le  \left(p^{-1}-1\right)\frac{2\mu r}{n}\left\Vert Z\right\Vert _{F},
\end{eqnarray*}
which means $\left\Vert \mathbb{E}\left[\sum_{(i,j)\in\Gammad}\mathcal{S}_{ij}^{2}\right]\right\Vert \le\frac{2\mu r}{np}$.
When $\epsilon_1 \ge \max\left\{\sqrt{\frac{32\beta\mu r \log n }{3np}},\frac{32\beta \mu r \log n }{3np}\right\}$, we apply Lemma \ref{lem:bernstein} and obtain
\begin{eqnarray*}
\mathbb{P}\left[\left\Vert \sum\mathcal{S}_{ij}^{2}\right\Vert \ge\epsilon_{1}\right]
 & \le & 2n^{2-2\beta}.
\end{eqnarray*}
Therefore, $\left\Vert p^{-1}\PP_{\T}\PP_{\Omega_{0}\cap\Gammao}\PP_{\T}-\PP_{\T}\PP_{\Gammao}\PP_{\T}\right\Vert <\epsilon_{1}$
w.h.p., which proves the first part of the lemma. On the other hand, when $\Gammao=\Gammad$, Lemma \ref{lem:det_Fro_bound} gives
\begin{eqnarray*}
\left\Vert \PP_{\T}\PP_{\Gammad}\PP_{\T}-\PP_{\T}\right\Vert  & = & \max_{Z:\left\Vert Z\right\Vert _{F}=1}\left\Vert \left(\PP_{\T}\PP_{\Gammad}\PP_{\T}-\PP_{\T}\right)Z\right\Vert _{F}\\
 & \le & \max_{Z:\left\Vert Z\right\Vert _{F}=1}\alpha\left\Vert \PP_{\T}Z\right\Vert _{F}
  \le  \alpha.
\end{eqnarray*}
The second part of  lemma then follows from the triangular inequality.
\end{proof}

The next three lemmas bound the norms of certain random matrices. Their proofs follow the same spirit as Lemma \ref{lem:OP1} by decomposing the random component into the sum of independent, bounded variables with small second moments, and then invoking Lemma \ref{lem:bernstein}. The following lemma is a generalization of \cite[Theorem 6.3]{CanRec}.
\begin{lemma}
\label{lem:OP_INF}
Suppose $\Omega_{0}$ is a set of indices obeying
$\Omega_{0}\sim$Ber$(p)$, $\Gammao$ is a fixed set of indices, and $Z$ is a fixed $n\times n$ matrix.
\begin{enumerate}
  \item For any $\beta>1$, we have
    \[
    \left\Vert \frac{1}{p}\PP_{\Omega_{0}\cap\Gammao}Z-\PP_{\Gammao}Z\right\Vert \le \sqrt{\frac{8\beta n\log n}{3p}}\left\Vert \PP_{\Gammad}Z\right\Vert _{\infty}
    \]
    with probability at least $1-2n^{1-\beta}$ provided $p\ge\frac{8\beta\log n}{3n}$.
  \item If in addition, $\Gammao = \Gammad$, where $\Gammad$ satisfies the assumptions in Theorem \ref{thm:random}, we have
    \[
    \left\Vert \frac{1}{p}\PP_{\Omega_{0}}Z-Z\right\Vert \le\left(\sqrt{\frac{8\beta n\log n}{3p}}+ d \right)\left\Vert Z\right\Vert _{\infty}
    \]
    with the same probability.
\end{enumerate}
\end{lemma}
\begin{proof}
For $(i,j)\in \Gammao$ define the random variable $\delta_{ij}=\mathbf{1}_{\{(i,j)\in\Omega_{0}\}}$.
Notice that
\[
\frac{1}{p}\PP_{\Omega_{0}\cap\Gammao}Z-\PP_{\Gammao}Z = \sum_{(i,j)\in\Gammao}(p^{-1}\delta_{ij}-1)Z_{i,j}\left(e_{i}e_{j}^{\top}\right) \triangleq \sum_{(i,j)\in\Gammao}\Xi_{ij}.
\]
Here $\Xi_{ij} \in \R^{n\times n}$ satisfies $\mathbb{E}\left[\Xi_{ij}\right]=0$, $\left\Vert \Xi_{ij}\right\Vert \le p^{-1}\left\Vert \PP_{\Gammao}Z\right\Vert _{\infty}$
and
\begin{eqnarray*}
 & & \left\Vert \mathbb{E}\left[\sum_{(i,j)\in \Gammao}\Xi_{ij}\Xi_{ij}^{\top}\right]\right\Vert = \left(p^{-1}-1\right)\left\Vert \sum_{(i,j)\in\Gammao}Z_{i,j}^{2}e_{i}e_{i}^{\top}\right\Vert \\
 & \le & \left(p^{-1}-1\right)\left\Vert \textrm{diag}\left(\sum_{(1,j)\in\Gammao}Z_{1,j}^{2},\ldots,\sum_{(n,j)\in\Gammao}Z_{n,j}^{2}\right)\right\Vert \\
 & \le & \left(p^{-1}-1\right)n\left\Vert \PP_{\Gammao}Z\right\Vert _{\infty}^{2}\le p^{-1}n\left\Vert \PP_{\Gammao}Z\right\Vert _{\infty}^{2}.
\end{eqnarray*}
A similar calculation yields $\left\Vert \mathbb{E}\left[\sum_{(i,j)\in \Gammao}\Xi_{ij}^{\top}\Xi_{ij}\right]\right\Vert \le p^{-1}n\left\Vert \PP_{\Gammao}Z\right\Vert _{\infty}^{2}$

When $p\ge\frac{8\beta\log n}{3n}$, we apply Lemma \ref{lem:bernstein}
and obtain
\begin{eqnarray*}
 &  & \mathbb{P}\left[\left\Vert \sum_{(i,j)\in \Gammao}\Xi_{ij}\right\Vert \ge\sqrt{\frac{8\beta n\log n}{3p}}\left\Vert \PP_{\Gammao}Z\right\Vert _{\infty}\right]\\
 & \le & 2n\exp\left(-\frac{3}{8}\cdot\frac{\frac{8\beta n\log n}{3p}\left\Vert \PP_{\Gammao}Z\right\Vert _{\infty}^{2}}{\frac{n}{p}\left\Vert \PP_{\Gammao}Z\right\Vert _{\infty}^{2}}\right)
 \le  2n^{1-\beta}.
 \end{eqnarray*}
Therefore, $\left\Vert \frac{1}{p}\PP_{\Omega_{0}\cap\Gammao}Z-\PP_{\Gammao}Z\right\Vert \le\sqrt{\frac{8\beta n\log n}{3p}}\left\Vert \PP_{\Gammao}Z\right\Vert _{\infty}$
w.h.p., which proves the first part of the lemma. On the other hand, when $\Gammao=\Gammad$, \cite[Proposition 3]{cspw} gives
\begin{eqnarray*}
 &  & \left\Vert \PP_{\Gammad}Z-Z\right\Vert
 = \left\Vert \PP_{\Gammad^{c}}Z\right\Vert
 \le d\left\Vert \PP_{\Gammad^{c}}Z\right\Vert _{\infty}.
 \end{eqnarray*}
The second part of the lemma then follows from the triangle inequality.
\end{proof}

The following lemma is a generalization of \cite[Lemma 3.1]{CanLiMaWri}.
\begin{lemma}
\label{lem:INF}Suppose $\Omega_{0}$ is a set of indices obeying
$\Omega_{0}\sim$Ber$(p)$, $\Gammao$ is a fixed set of indices, and $Z$ is a fixed $n\times n$
matrix in $\T$.
\begin{enumerate}
    \item For any $\beta>1$ and $\epsilon_3<1$, we have
        \[
        \left\Vert \frac{1}{p}\PP_{\T}\PP_{\Omega_{0}\cap\Gammao}\PP_{\T}Z-\PP_{\T}\PP_{\Gammao}\PP_{\T}Z \right\Vert _{\infty} \le \epsilon_3\left\Vert Z\right\Vert _{\infty}
        \]
        with probability at least $1-2n^{2-2\beta}$ provided $p\ge\frac{32\beta\mu r\log n}{3n\epsilon_3^2}$.
    \item If in addition, $\Gammao=\Gammad$, where $\Gammad$ satisfies the assumptions in Theorem \ref{thm:random}, we have
        \[
        \left\Vert \frac{1}{p}\PP_{\T}\PP_{\Omega_{0}}\PP_{\T}Z-Z\right\Vert _{\infty}\le\left(\epsilon_3+\alpha\right)\left\Vert Z\right\Vert _{\infty}
        \]
        with the same probability.
\end{enumerate}
\end{lemma}
\begin{proof}
For $(i,j)\in \Gammao$, set $\delta_{ij}=\mathbf{1}_{\{(i,j)\in\Omega_{0}\}}$. Fix $(a,b)\in [n]\times[n]$. Notice that
\begin{eqnarray*}
& &\left(\frac{1}{p}\PP_{\T}\PP_{\Omega_{0}\cap\Gammao}\PP_{\T}Z-\PP_{\T}\PP_{\Gammao}\PP_{\T}Z\right)_{a,b} \\
 & = &\sum_{(i,j)\in\Gammao}\left\langle (p^{-1}\delta_{ij}-1)Z_{i,j}\PP_{\T}(e_{i}e_{j}^{\top}),\; e_{a}e_{b}^{\top}\right\rangle \triangleq\sum_{(i,j)\in\Gammao}\xi_{ij}
\end{eqnarray*}
where $\mathbb{E}\left[\xi_{ij}\right]=0$. For $(i,j)\in\Gammao$,
we have
\begin{eqnarray*}
 \left|\xi_{ij}\right| & \le & p^{-1}\left\Vert \PP_{\T}(e_{i}e_{j}^{\top})\right\Vert _{F}\left\Vert \PP_{\T}(e_{a}e_{b}^{\top})\right\Vert _{F}\left|Z_{i,j}\right|\\
 & \le & \frac{2\mu r}{np}\left\Vert \PP_{\Gammao}Z\right\Vert _{\infty}
\end{eqnarray*}
The second moment is bounded by
\begin{eqnarray*}
 & & \left|\mathbb{E}\left[\sum_{(i,j)\in\Gammao}\xi_{ij}^{2}\right]\right| \\
 & = & \left|\sum_{(i,j)\in\Gammao}\mathbb{E}\left[(p ^{-1}\delta_{ij}-1)^{2}\right]\left\langle \PP_{\T}(e_{i}e_{j}^{\top}),\; e_{a}e_{b}^{\top}\right\rangle ^{2}Z_{i,j}^{2}\right|\\
 & \le & \left(p ^{-1}-1\right)\left\Vert \PP_{\Gammao}Z\right\Vert _{\infty}^{2}\sum_{(i,j)\in\Gammao}\left\langle e_{i}e_{j}^{\top},\; \PP_{\T}(e_{a}e_{b}^{\top})\right\rangle ^{2}\\
 & = & \left(p ^{-1}-1\right)\left\Vert \PP_{\Gammao}Z\right\Vert _{\infty}^{2}\left\Vert \PP_{\Gammao}\PP_{\T}(e_{a}e_{b}^{\top})\right\Vert _{F}^{2}\\
 & \le & \left(p ^{-1}-1\right)\frac{2\mu r}{n}\left\Vert \PP_{\Gammao}Z\right\Vert _{\infty}^{2}
 \le \frac{2\mu r}{np }\left\Vert \PP_{\Gammao}Z\right\Vert _{\infty}^{2}.
\end{eqnarray*}
When $p \ge\frac{32\beta\mu r\log n}{3n\epsilon_3^2}$ and $\epsilon_3<1$, we apply Lemma \ref{lem:bernstein} and obtain
\begin{eqnarray*}
 &  & \mathbb{P}\left[\left|\left(\frac{1}{p}\PP_{\T}\PP_{\Omega_{0}\cap\Gammao}\PP_{\T}Z-\PP_{\T}\PP_{\Gammao}\PP_{\T}Z\right)_{a,b}\right|\ge\epsilon_3\left\Vert \PP_{\Gammao}Z\right\Vert _{\infty}\right]\\
 & \le & 2\exp\left(-\frac{3(\epsilon_3)^{2}\left\Vert \PP_{\Gammao}Z\right\Vert _{\infty}^{2}}{8\cdot\frac{2\mu r}{np}\left\Vert \PP_{\Gammao}Z\right\Vert _{\infty}^{2}}\right)
 \le 2n^{-2\beta}.
 \end{eqnarray*}
Union bound then yields
\begin{eqnarray*}
\left\Vert \frac{1}{p}\PP_{\T}\PP_{\Omega_{0}\cap\Gammao}\PP_{\T}Z-\PP_{\T}\PP_{\Gammao}\PP_{\T}Z\right\Vert _{\infty} \le \epsilon_3\left\Vert \PP_{\Gammao}Z\right\Vert _{\infty}
\end{eqnarray*}
with high probability, which proves the first part of the lemma. On the other hand, when $\Gammao=\Gammad$, by \eqref{eq:upper_bound} we have $\left\Vert \PP_{\T}\PP_{\Gammad}\PP_{\T}Z-Z\right\Vert _{\infty}=\left\Vert \PP_{\T}\PP_{\Gammad^{c}}Z\right\Vert _{\infty}\le\alpha\left\Vert Z\right\Vert _{\infty}$.
The second part of the lemma then follows from triangle inequality.
\end{proof}

The next two lemmas bound $\left\Vert \PP_{\T}\EE\right\Vert _{\infty}$.
\begin{lemma}
\label{lem:PTE_infd} Under the assumption of Theorem \ref{thm:random}, we have
\[
\left\Vert \PP_{\T}\PP_{\Omegad}\EE\right\Vert _{\infty} \le \alpha.
\]
\end{lemma}
\begin{proof}
By assumption $\Omegad$ contains at most $d$ entries from each row/column, so repeating the proof of Lemma \ref{lemma:zero-intersection} yields the desired bound.
\end{proof}

\begin{lemma}
\label{lem:PTE_infr} Under the assumption of Theorem \ref{thm:random} and conditioned on $\Omegar$, we have
\[
\left\Vert \PP_{\T}\PP_{\Omegar\backslash\Omegad}\EE\right\Vert _{\infty} \le C\sqrt{\frac{\mu r}{n}p_0 \log n}
\]
with high probability for some constant $C>0$.
\end{lemma}
\begin{proof}
Set $E=\PP_{\Omegar\backslash\Omegad}\EE$; observe that each entry of $E$ in $(\Omegar\cap\OmegaOd)\backslash\Omegad$ is non-zero with probability $p_{0}$ and has random sign, independent of each other. Since we have
\begin{eqnarray*}
& & \left\Vert \PP_{\T}E\right\Vert _{\infty}  =  \left\Vert \PP_{U}E+P_{V}E-P_{U}\PP_{V}E\right\Vert _{\infty}\\
 & \le & \left\Vert UU^{\top}E\right\Vert _{\infty}+\left\Vert EVV^{\top}\right\Vert _{\infty}+\left\Vert UU^{\top}EVV^{\top}\right\Vert _{\infty},
\end{eqnarray*}
it suffices to bound these three terms.
From the incoherence property of $U$, we know
\begin{eqnarray*}
\left\Vert UU^\top \right\Vert_{\infty} = \max_{i,j} \left\vert e_i^\top UU^\top e_j \right\vert
\le \frac{\mu r}{n},
\end{eqnarray*}
and
\begin{eqnarray*}
\left\Vert e_i^\top UU^\top \right\Vert^{2}
 \le \frac{\mu r}{n},\quad \forall i
\end{eqnarray*}
Now we bound $\left\Vert UU^{\top}E\right\Vert _{\infty}$. For simplicity,
we focus on the $(1,1)$ entry of $\left(UU^{\top}E\right)$ and denote
it as $X$. Set $s^\top = e_1^\top UU^\top$. Observe that $X=\sum_{i:(i,1)\in(\Omegar\cap\OmegaOd)\backslash\Omegad}s_{i} E_{i,1}$ , $E_{i,1}$'s
are i.i.d., with $\mathbb{E}\left[s_{i} E_{i,1}\right]  =  0$ and
\begin{eqnarray*}
\left|s_{i}^\top E_{i,1}\right| & \le & \left|s_{i} \right|\le\frac{\mu r}{n},\quad\textrm{a.s.}\\
\textrm{Var}\left(X\right) & = & \sum_{i:(i,1)\in(\Omegar\cap\OmegaOd) \backslash \Omegad}(s_{i})^{2} p_{0}\le\frac{\mu r}{n}p_0.
\end{eqnarray*}
Standard bernstein inequality \eqref{eq:bernstein0}
thus gives
\[
\mathbb{P}\left[\left|X\right|>t\right] \le 2\exp\left(-\frac{t^2}{2\frac{\mu r}{n}p_0 + \frac{2\mu r}{3n}t}\right).
\]
Under the assumption of Theorem \ref{thm:random}, we can choose $t=C\max\{\frac{\mu r}{n}\log n,\; \sqrt{\frac{\mu r}{n}p_0  \log n} \}$ for some $C$ sufficiently large and apply the union bound to obtain
\[
\left\Vert UU^{\top}E\right\Vert _{\infty} \le C\max\left\{\frac{\mu r}{n}\log n,\; \sqrt{\frac{\mu r}{n}p_0 \log n}\right\},\quad\textrm{w.h.p.}
\]
Similarly, $\left\Vert EVV^{\top}\right\Vert _{\infty}$ is also bounded by the right hand side of the above equation. Finally, denote $w:=VV^{\top}e_1$
and observe that
\begin{eqnarray*}
\left(UU^{\top}EVV^{\top}\right)_{1,1} & = & \sum_{(i,j)\in(\Omegar\cap\OmegaOd)\backslash\Omegad}s_{i}w_{j}E_{i,j}.
\end{eqnarray*}
Then a similar application of Bernstein inequality and the union bound
gives
\[
\left\Vert UU^{\top}EVV^{\top}\right\Vert _{\infty} \le C'\max\left\{\frac{\mu^2 r^2}{n^2}\log n,\; \frac{\mu r}{n} \sqrt{p_0 \log n}\right\},\quad\textrm{w.h.p.}
\]
The lemma follows from observing that $\frac{\mu r}{n}\le 1$ and $p_0\ge \frac{\mu r \log n}{n}$ under the assumptions of Theorem \ref{thm:random}.
\end{proof}

Finally, we prove Lemma \ref{lem:exceptionterm}.
\begin{IEEEproof} (of Lemma \ref{lem:exceptionterm})
Recall that by definition $\Omegar^{c}=\bigcup_{k=1}^{k_{0}}\Omegak$, so
we have
\begin{eqnarray*}
 &  & \lambda\left(\frac{1}{q_{1}}\PP_{\Omegaone\cap\Gammad}-\mathcal{I}\right)\PP_{\T}\PP_{\Omegar\backslash\Omegad}(\EE )\\
 & = & \lambda\left(\PP_{\Gammad}-\mathcal{I}\right)\PP_{\T}\PP_{\Omegar\backslash\Omegad}(\EE )+\lambda\left(\frac{1}{q_{1}}\PP_{\Omegaone}-\mathcal{I}\right)\PP_{\Gammad}\PP_{\T}\PP_{\Omegar\backslash\Omegad}(\EE )\\
 & = & \lambda\left(\PP_{\Gammad}-\mathcal{I}\right)\PP_{\T}\PP_{\Omegar\backslash\Omegad}(\EE )+\lambda\left(\frac{1}{q_{1}}\PP_{\Omegaone}-\mathcal{I}\right)\PP_{\Gammad}\PP_{\T}\PP_{\left(\Omegaone\right)^{c}}\PP_{\left(\Omegatwo\right)^{c}\bigcap\ldots\bigcap\left(\Omegako\right)^{c}\backslash\Omegad}\left(\EE \right)\\
 & \triangleq & \lambda\left(\PP_{\Gammad}-\mathcal{I}\right)\PP_{\T}\PP_{\Omegar\backslash\Omegad}(\EE )+\lambda\left(\frac{1}{q_{1}}\PP_{\Omegaone}-\mathcal{I}\right)\PP_{\Gammad}\PP_{\T}\PP_{\left(\Omegaone\right)^{c}}E
\end{eqnarray*}
where $E$ is a matrix with independent random signed entries supported
on $\OmegaO\cap\Omega^{(2)c}\cap\cdots\cap\Omega^{(k_{0})c}\backslash\Omegad$.
The operator norm of the first term is bounded using \cite[Proposition 3]{cspw} and Lemma \ref{lem:PTE_infr} as
\[
d\left\Vert \lambda \PP_{\T}\PP_{\Omegar\backslash\Omegad}(\EE )\right\Vert _{\infty}\le d \cdot C\lambda\sqrt{\frac{\mu r}{n} p_{0}\log n}\le C'.
\]
Let $\delta_{ab}=\mathbf{1}_{\left\{(a,b)\in\Omegaone\right\}}$, then
the second term can be decomposed as
\begin{eqnarray*}
 &   & \lambda\left(\frac{1}{q_{1}}\PP_{\Omegaone}-\mathcal{I}\right)\PP_{\Gammad}\PP_{\T}\PP_{\left(\Omegaone\right)^{c}}E \\
 & = & \sum_{a,b,a',b'}\lambda\left(\frac{1}{q_{1}}\delta_{ab}-1\right)\left(1-\delta_{a'b'}\right)E_{a',b'}\left\langle \PP_{\Gammad}\PP_{\T}(e_{a'}e_{b'}^{\top}),e_{a}e_{b}^{\top}\right\rangle e_{a}e_{b}^{\top}\\
 & = & \sum_{(a',b')=(a,b)}+\sum_{(a',b')\neq(a,b)}
\end{eqnarray*}
We bound the operator norm of the above two terms separately.

The
diagonal term is bounded as
\begin{eqnarray*}
\left\Vert \sum_{(a',b')=a,b}\right\Vert  & = & \left\Vert \sum_{a,b}\lambda\left(\delta_{ab}-1\right)E_{a,b}\left\langle \PP_{\Gammad}\PP_{\T}(e_{a}e_{b}^{\top}),e_{a}e_{b}^{\top}\right\rangle e_{a}e_{b}^{\top}\right\Vert \\
 & \le & \left\Vert \sum_{a,b}\lambda\left(\delta_{ab}-q_{1}\right)X_{a,b}e_{a}e_{b}^{\top}\right\Vert +\left\Vert \sum_{a,b}\lambda\left(q_{1}-1\right)X_{a,b}e_{a}e_{b}^{\top}\right\Vert \\
 & = & q_{1}\left\Vert \lambda\left(\frac{1}{q_{1}}\PP_{\Omegaone}-\mathcal{I}\right)X\right\Vert +\left\Vert \sum_{a,b}\lambda\left(q_{1}-1\right)X_{a,b}e_{a}e_{b}^{\top}\right\Vert
\end{eqnarray*}
where $X_{a,b}=E_{a,b}\left\langle \PP_{\Gammad}\PP_{\T}(e_{a}e_{b}^{\top}),e_{a}e_{b}^{\top}\right\rangle $.
The first part of Lemma \ref{lem:OP_INF} with $\Omega_0=\Omegaone$ and $\Gammao=[n]\times[n]$ bounds the first term by
$q_{1}\lambda C\sqrt{\frac{n\log n}{q_{1}}}\left\Vert X\right\Vert _{\infty} \le q_{1}\lambda C\sqrt{\frac{n\log n}{q_{1}}} \frac{2\mu r}{n} \le C'$.
We then apply \cite[Lemma 6.4]{CanRec} and a standard bound of the operator
norm of a random matrix to bound the second term by $\lambda(1-q_{1})\frac{2\mu r}{n}\left\Vert E\right\Vert \le\lambda\frac{2\mu r}{n}\cdot\sqrt{np_{0}\log n}\le C$.

The off-diagonal term can be expressed as
\begin{eqnarray*}
 \sum_{(a',b')\neq(a,b)} & = & \sum_{(a',b')\neq(a,b)}\lambda\left(\frac{1}{q_{1}}\delta_{a'b'}-1\right)\left(1-\delta_{ab}\right)E_{ab}\left\langle \PP_{\Gammad}\PP_{\T}(e_{a}e_{b}^{\top}),e_{a'}e_{b'}^{\top}\right\rangle e_{a'}e_{b'}^{\top}\\
 & = & \frac{\lambda}{q_{1}}\sum_{(a',b')\neq(a,b)}\left(\delta_{a'b'}-q_{1}\right)\left(q_{1}-\delta_{ab}\right)E_{ab}\left\langle \PP_{\Gammad}\PP_{\T}(e_{a}e_{b}^{\top}),e_{a'}e_{b'}^{\top}\right\rangle e_{a'}e_{b'}^{\top}\\
 &   & \qquad +\frac{\lambda}{q_{1}}\sum_{(a',b')\neq(a,b)}\left(\delta_{a'b'}-q_{1}\right)\left(1-q_{1}\right)E_{ab}\left\langle \PP_{\Gammad}\PP_{\T}(e_{a}e_{b}^{\top}),e_{a'}e_{b'}^{\top}\right\rangle e_{a'}e_{b'}^{\top}
\end{eqnarray*}
The operator norm of first term can be bounded using the decoupling
argument in \cite{CanRec}. In particular, we can repeat the proof of \cite[Lemma 6.7]{CanRec} with $p=q_{1}$, $\xi_{ab}=\delta_{ab}-q_{1}$ and $\left\Vert \PP_{\Gammad}\PP_{\T}(e_{a}e_{b}^{\top})\right\Vert _{F}^{2}\le\frac{2\mu r}{n}$
to bound the first term by $C'\lambda\sqrt{\mu r}\log n\left\Vert E\right\Vert _{\infty}\le C$.
Let $H_{a'b'}=\sum_{a,b:(a,b)\neq(a',b')}E_{a,b}\left\langle \PP_{\Gammad}\PP_{\T}(e_{a}e_{b}^{\top}),e_{a'}e_{b'}^{\top}\right\rangle $;
the second term can be bounded as
\begin{eqnarray}
 &  & \left\Vert \frac{\lambda(1-q_{1})}{q_{1}}\sum_{(a',b')\neq(a,b)}\left(\delta_{a'b'}-q_{1}\right)E_{a,b}\left\langle \PP_{\Gammad}\PP_{\T}(e_{a}e_{b}^{\top}),e_{a'}e_{b'}^{\top}\right\rangle e_{a'}e_{b'}^{\top}\right\Vert \nonumber \\
 & \le & \lambda\left\Vert \sum_{a',b'}\left(\frac{1}{q_{1}}\delta_{a'b'}-1\right)\sum_{(a,b):(a,b)\neq (a',b')}E_{a,b}\left\langle \PP_{\Gammad}\PP_{\T}(e_{a}e_{b}^{\top}),e_{a'}e_{b'}^{\top}\right\rangle e_{a'}e_{b'}^{\top}\right\Vert \nonumber \\
 & = & \lambda\left\Vert \left(\frac{1}{q_{1}}\PP_{\Omegaone}-\mathcal{I}\right)H\right\Vert \nonumber \\
 & \le & \lambda C\sqrt{\frac{n\log n}{q_{1}}}\left\Vert H\right\Vert _{\infty}\label{eq:x}
\end{eqnarray}
where we use the first part of Lemma \ref{lem:OP_INF} with $\Omega_0=\Omegaone$ and $\Gammao=[n]\times[n]$ in the inequality. Further
observe that
\begin{eqnarray*}
H_{a',b'} & = & \left(\sum_{a,b}E_{a,b}\left\langle \PP_{\Gammad}\PP_{\T}(e_{a}e_{b}^{\top}),e_{a'}e_{b'}^{\top}\right\rangle \right)-E_{a',b'}\left\langle \PP_{\Gammad}\PP_{\T}(e_{a'}e_{b'}^{\top}),e_{a'}e_{b'}^{\top}\right\rangle \\
 & = & \left(\PP_{\Gammad}\PP_{\T}E\right)_{a',b'}-E_{a',b'}\left\langle \PP_{\Gammad}\PP_{\T}(e_{a'}e_{b'}^{\top}),e_{a'}e_{b'}^{\top}\right\rangle
\end{eqnarray*}
so we have
\begin{eqnarray*}
\left\Vert H\right\Vert _{\infty} & \le & \left\Vert P_{\T}E\right\Vert _{\infty}+\left\Vert E\right\Vert _{\infty}\left\Vert \PP_{\T}(e_{a'}e_{b'}^{\top})\right\Vert _{F}^{2}\\
 & \le & C\sqrt{\frac{\mu r}{n}p_{0}\log n}+\frac{2\mu r}{n} \le C'\sqrt{\frac{\mu r}{n}p_{0}\log n}.
\end{eqnarray*}
where we use Lemma \ref{lem:PTE_infr}. It follows that the right hand side of (\ref{eq:x})
is bounded by $\lambda\sqrt{\frac{n\log n}{q_{1}}}C'\sqrt{\frac{\mu r}{n}p_{0}\log n}\le C''$.
This completes the proof of the lemma.
\end{IEEEproof}


\end{document}